\documentclass[11pt]{article}

\usepackage[utf8]{inputenc}
\usepackage[T1]{fontenc}
\usepackage{lmodern}
\usepackage[letterpaper,margin=0.90in]{geometry}

\usepackage{makeidx}
\usepackage{graphicx,comment}
\usepackage[linesnumbered,vlined]{algorithm2e}
\usepackage[textsize=tiny]{todonotes}

\usepackage{amsmath}
\usepackage{amsthm}
\usepackage{amsfonts, amssymb}
\usepackage{cite}
\usepackage{fullpage}
\usepackage{framed}
\usepackage{epstopdf}
\usepackage[framemethod=tikz]{mdframed}
\usepackage{booktabs}

\usepackage[nameinlink,capitalize]{cleveref}

\newtheorem{theorem}{Theorem}[section]
\newtheorem{lemma}[theorem]{Lemma}

\newtheorem{observation}[theorem]{Observation}
\newtheorem{claim}[theorem]{Claim}
\newtheorem{definition}[theorem]{Definition}

\newcommand{\eps}{\varepsilon}

\newenvironment{lemma-repeat}[1]{\begin{trivlist}
\item[\hspace{\labelsep}{\bf\noindent Lemma \ref{#1} }]\em }%
{\end{trivlist}}
\newenvironment{theorem-repeat}[1]{\begin{trivlist}
\item[\hspace{\labelsep}{\bf\noindent Theorem \ref{#1} }]\em }%
{\end{trivlist}}

\newcommand*\samethanks[1][\value{footnote}]{\footnotemark[#1]}

\newcommand{\Ltop}{L_{top}}
\newcommand{\Wmax}{W}
\newcommand{\InIS}{\texttt{InIS}}
\newcommand{\NotInIS}{\texttt{NotInIS}}
\newcommand{\weightUpdate}{weightUpdate}
\newcommand{\removed}{removed}
\newcommand{\addedToIS}{addedToIS}
\newcommand{\Wait}{\texttt{Wait}}
\newcommand{\Dvi}{D_{v,i}}
\newcommand{\Dui}{D_{u,i}}
\newcommand{\Dnvi}{D_{N(v),i}}
\newcommand{\MIS}{\texttt{MIS}(G)}

\usepackage{times}

\begin{document}
\begin{titlepage}
	\title{Distributed Approximation of Maximum Independent Set and Maximum Matching }
	\author{Reuven Bar-Yehuda\thanks{Technion, Department of Computer Science, \texttt{\{reuven, ckeren\}@cs.technion.ac.il, gregory.schwartzman@gmail.com}. Supported in part by the Israel Science Foundation (grant 1696/14).}
		\and Keren Censor-Hillel\samethanks \and Mohsen Ghaffari\thanks{ETH Zurich, \texttt{ghaffari@mit.edu}.} \and Gregory Schwartzman\samethanks[1]}
		\date{}
	\maketitle
	
	\begin{abstract}
		We present a simple distributed $\Delta$-approximation algorithm for maximum weight independent set (MaxIS) in the $\mathsf{CONGEST}$ model which completes in $O(\MIS\cdot \log\Wmax)$ rounds, where $\Delta$ is the maximum degree, $\MIS$ is the number of rounds needed to compute a maximal independent set (MIS) on $G$, and $\Wmax$ is the maximum weight of a node. 
		Plugging in the best known algorithm for MIS gives a randomized solution in $O(\log n \log\Wmax)$ rounds, where $n$ is the number of nodes.
		We also present a deterministic $O(\Delta +\log^* n)$-round algorithm based on coloring.
		
		We then show how to use our MaxIS approximation algorithms to compute a $2$-approximation for maximum weight matching without incurring any additional round penalty in the $\mathsf{CONGEST}$ model. We use a known reduction for simulating algorithms on the line graph while incurring congestion, but we show our algorithm is part of a broad family of \emph{local aggregation algorithms} for which we describe a mechanism that allows the simulation to run in the $\mathsf{CONGEST}$ model without an additional overhead.
		
		Next, we show that for maximum weight matching, relaxing the approximation factor to ($2+\eps$) allows us to devise a distributed algorithm requiring $O(\frac{\log \Delta}{\log\log\Delta})$ rounds for any constant $\eps>0$. For the unweighted case, we can even obtain a $(1+\eps)$-approximation in this number of rounds. These algorithms are the first to achieve the
provably optimal round complexity with respect to dependency on $\Delta$.
	\end{abstract}
	\thispagestyle{empty}
	
\end{titlepage}

\section{Introduction}

We address the fundamental problems of approximating the maximum independent set and the maximum matching of a graph in the classic distributed $\mathsf{CONGEST}$ model~\cite{Peleg:2000}. In this model, the $n$ nodes of the graph communicate in synchronous rounds, by sending one $O(\log n)$-bit message per round along links of the graph. Table~\ref{tab: results} summarizes our contributions. Below, we elaborate on our results, the challenges, and how we overcome them.
\begin{table}[htbp]
\centering
%
%
%
%
%
%

\begin{tabular*}{\linewidth}{@{}l@{\extracolsep{\fill}}c@{}c@{}c@{}r@{}}
	\toprule
	Problem       &  Approximation & Prev. Results  & Our Results & notes \\
	\midrule
	
	MaxIS~|~MWM    &
	$\Delta$~|~2&
	---&
	$O(\MIS \log\Wmax )$ &
	rand.\\
	
	MaxIS~|~MWM    &
	$\Delta$~|~2&
	--- &
	$O(\Delta +\log^* n )$  &
	det.\\

	MWM  &
	$2+\eps$ &
	$O(\log n)$ &
	$O(\log\Delta/\log\log\Delta)$ &
	rand.\\
	
	MCM  &
	$1+\eps$ &
	$O(\log n)$ &
	$O(\log\Delta/\log\log\Delta)$ &
	rand.\\
	
%
	
	
	\bottomrule
\end{tabular*}
\caption{Summary of results for the $\mathsf{CONGEST}$ model. Here $n$ denotes the number of nodes, $\Delta$ is their maximum degree, and $W$ is the maximum weight.
}
\label{tab: results}
\end{table}

\vspace{-14pt}
\subsection{Our Results, Part I: Better Approximations}
\paragraph{$\Delta$-approximation algorithms for maximum weight independent set.} We present a simple distributed $\Delta$-approximation algorithm for maximum weight independent set (MaxIS), where $\Delta$ is the maximum degree, which completes in $O(\MIS\cdot \log\Wmax)$ rounds, where $\MIS$ is the number of rounds needed to compute a maximal independent set (MIS) on $G$, and $\Wmax$ is the maximum weight of a node. As standard, we assume that $\Wmax$ is at most polynomial in $n$, so that the weight of each edge can be described in one message. Our algorithm adapts the \emph{local ratio} technique~\cite{BarYehudaE1985} for maximization problems \cite{bar2001unified} to the distributed setting in a novel, yet simple, manner. Roughly speaking, in the simplest form of this technique, one repeatedly picks a node $v$ and reduces its weight from every $u\in N(v)$, where $N(v)$ is the set of neighbors of $v$. Every neighbor $u\in N(v)$ whose weight becomes less than or equal to zero is removed from the graph, while $v$ is added to a stack. We repeat this process with the induced graph until no nodes remain. We then begin popping nodes from the stack, adding them to the independent set if they have no neighbors in the set. This yields a $\Delta$-approximation.

The challenge in translating this framework to the distributed setting is that if we allow all nodes to perform weight reductions simultaneously, then the above does not hold. For example, consider a star graph where the weight of the center is larger than the weight of any of its neighbors but smaller than their sum. After a single iteration the weights of all the nodes become negative, and no node gets selected. However, we show that if we first compute an independent set and then go on to perform weight reductions we achieve a $\Delta$-approximation factor, while allowing using the power of parallelism. At each iteration we find an MIS, and the nodes chosen to the MIS perform weight reductions. This process is repeated until no nodes with positive weight remain. Nodes are then added to the independent set in reverse order of removal while maintaining the independence constraints. To analyze the running time, our main technique is to group the nodes into $\log \Wmax$ \emph{layers} based on their weight. At each iteration, all of the nodes from the topmost layer move to lower layers.

This results in a round complexity of $O(\MIS\cdot\log \Wmax)$ in the $\mathsf{CONGEST}$ model. Our algorithm is deterministic apart from using a black-box algorithm to find an MIS at each iteration. Whether our algorithm is randomized or deterministic depends on the MIS algorithm it uses as a black-box.

We also present a deterministic coloring-based algorithm running in $O(\Delta+ \log^* n)$ rounds. Here we first color the graph using $\Delta+1$ colors, and then use each color group as an independent set to perform weight reductions as in the previous algorithm. 

\paragraph{$2$-approximation algorithms for maximum weighted matching.} We use a known reduction  to simulate algorithms on the line graph \cite{kuhn2005price}, our MaxIS $\Delta$-approximation algorithm gives a $2$-approximation for maximum weight matching. Simulating an execution on the line graph in a naive fashion results in a $O(\Delta)$ multiplicative overhead in the $\mathsf{CONGEST}$ model. We show our algorithm is part of a broad family of \emph{local aggregation} algorithms for which we describe a mechanism which allows the simulation to run in the $\mathsf{CONGEST}$ model without added overhead.

Our deterministic coloring-based algorithm has a favorable running time compared to the algorithm presented in \cite{even2015distributed} with parameters that result in a $2$-approximation. Our randomized algorithm improves upon the $(2+\epsilon)$-approximation factor of \cite{lotker2008improved}. Using the maximal matching algorithm of \cite{BarenboimEPS16} on the original graph as an MIS algorithm on the line graph we get a running time of $O((\log \Delta + \log^4\log n)\cdot \log \Wmax)$ \footnote{Note that $\Delta$ and $n$ are the parameters of the original graph and not the line graph.}, with high probability, for the $\mathsf{LOCAL}$ model, and using Luby's classical MIS algorithm\cite{luby1986simple}, we get an $O(\log n \cdot \log \Wmax)$ algorithm\footnote{Here the MIS algorithm is executed on the line graph, so we get $O(\MIS) = O(\log n^2) = O(\log n)$.} for the $\mathsf{CONGEST}$ model. For constant values of $\Wmax$, this is $O(\log n)$ rounds.

\subsection{Our Results, Part II: Faster Approximations}
\paragraph{Approximations with Optimal Time-Complexity:} We provide two approximations algorithms for maximum matching that achieve the optimal round complexity of $O(\log\Delta/\log \log \Delta)$: The first achieves a $(2+\epsilon)$-approximation of maximum weight matching, and the second a $(1+\epsilon)$-approximation of maximum cardinality matching, for any constant $\eps>0$.


These two algorithms improve upon the $O(\log n)$-round algorithms of Lotker et al.\cite{lotkerMatchingImproved} for the same problems and same approximation guarantees. Furthermore, these two algorithms are the first constant-approximation algorithms that achieve an optimal round complexity, matching the $\Omega(\log\Delta/\log \log \Delta)$ lower bound of Kuhn et al.~\cite{kuhn2006price}. We note that this lower bound holds for any constant approximation, and so long as $\log\Delta \leq \sqrt{\log n}$. 

\paragraph{Method Outline}
A key ingredient in both of the above fast algorithms is an improvement of the \emph{nearly-maximal independent set} algorithm of Ghaffari~\cite{Ghaffari2016}. A nearly-maximal independent set is an independent set for which each node in the graph is in the set or has a neighbor in the set with probability at least $1-\delta$ for a small $\delta$.
The main result of \cite{Ghaffari2016} is a maximal independent set algorithm with round complexity of $O(\log\Delta)+2^{O(\sqrt{\log\log n})}$. The central building block in that result was finding a nearly-maximal independent set in $O(\log\Delta)$ rounds. Here, we provide an improved nearly-maximal independent set algorithm with a round complexity of $O(\log\Delta/\log\log \Delta)$. This algorithm builds upon the techniques of \cite{Ghaffari2016}, but with some crucial modifications. The modification is partially inspired by the ideas of the recent vertex-cover approximation algorithm of Bar-Yehuda et al.~\cite{bar2016distributed}, of balancing two types of progresses. While this improvement does not allow us to improve upon Ghaffari's MIS algorithm, it helps us in obtaining our fast maximum matching approximation algorithms, as we discuss next. 

For the $(2+\epsilon)$-approximation, this improved nearly-maximal independent set algorithm is essentially enough. We run it on the line graph of the network graph, and argue that it gives an $(2+\epsilon)$-approximation of the maximum unweighted matching. To argue that the algorithm works in the $\mathsf{CONGEST}$ model, even when run on the line graph of the network graph, we use the property that this nearly-maximal independent set algorithm is a local aggregation algorithm. Then, we extend this approximation algorithm to the \emph{weighted} case, using techniques of~\cite{lotker2009distributed,lotkerMatchingImproved}.

For the unweighted $(1+\epsilon)$-approximation, our goal is to use the general framework of Hopcroft and Karp~\cite{HopcroftKarp1973}, in which we repeatedly search for short non-intersecting augmenting paths and augment the matching with them, hence improving its size. However, in our setting, this does not work as is and poses significant challenges. One key challenge is that, to have the desired approximation factor, we need a much stronger near-maximality guarantee. It does not suffice to have a low probability for each short augmenting path to remain; we need to show that each node has a low probability of having a remaining augmenting path. To overcome the obstacles, first we show how to find a nearly-maximal matching in low-rank hypergraphs and how to modify the algorithm for obtaining the ($1+\epsilon$)-approximation guarantee in the $\mathsf{LOCAL}$ model. 

Making the algorithm suitable for the $\mathsf{CONGEST}$ model is even more demanding, in part because here we cannot explicitly work with the structure of the intersections between short augmenting paths; instead, we need to have a new variant of the near-maximal independent set algorithm that works \emph{on the fly}. At a high level, we first address bipartite graphs, and show how to find a nearly-maximal independent set of short augmenting paths in them. Since the augmenting paths are not known explicitly, an interesting aspect here will be a variant of the dynamic probability adjustments in the algorithm of\cite{Ghaffari2016}. Now, various nodes of a path might decide differently regarding whether to raise or lower its probability. However, we will prove that still the net effect provides a sufficient move in the right direction. We complete by generalizing this from bipartite graphs to all graphs, using an idea of Lotker et al.~\cite{lotkerMatchingImproved}, which essentially transforms the problem into randomly chosen bipartite subgraphs of it.

\subsection{Related Work} 
The maximum independent set problem is known to be NP-hard, as it is complementary to the maximum clique problem, which is one of Karp's 21 NP-hard problems \cite{Karp72}.
In the sequential setting, an excellent summary of the known results is given by \cite{bansal2015approximating}, which we overview in what follows. For general graphs, the best known algorithm achieves a $O(n\log^2 \log n/ \log^3 n)$-approximation factor \cite{feige2004approximating}. Assuming $NP \nsubseteq ZPP$, \cite{haastad1996clique} shows that no $(n^{1-\epsilon})$-approximation exists for every constant $\epsilon > 0$.

When the degree is bounded by $\Delta$, a simple $(\Delta+2)/3$-approximation is achieved by greedily adding the node with minimal degree to the independent set and removing its neighbors \cite{halldorsson1997greed}. The best known approximation factor is $O(\Delta \log \log \Delta / \log \Delta)$ \cite{alon1998approximating, halldorsson1998approximations, halperin2002improved, halldorsson2000approximations, karger1998approximate}.
Conditioned on the Unique Games Conjecture, there exist a $\Omega(\Delta / \log^2 \Delta)$-approximation bound \cite{austrin2009inapproximability}, where $\Delta$ is constant or some mildly increasing function of $n$. Assuming $P\neq NP$, a bound of $\Omega(\Delta / \log ^4 \Delta)$ is given in \cite{chan2013approximation}.

As for the distributed case, \cite{lenzen2008leveraging, czygrinow2008fast} give a lower bound of $\Omega(\log^* n)$ rounds for any deterministic algorithm approximating MaxIS, while \cite{czygrinow2008fast} provide randomized and deterministic approximations for planar graphs. In \cite{bodlaenderbrief}, an $O(1/\epsilon)$-round $\mathsf{LOCAL}$ randomized algorithm for $O(n^{\epsilon})$-approximation is presented for the unweighted case, along with a matching lower bound.

Maximum matching is a classical optimization problem, for which the first polynomial time algorithm was given by Edmonds \cite{edmonds1965maximum, edmonds1965paths} for both the weighted and unweighted case. In the distributed setting, the first algorithm for computing an approximate maximum matching was given in \cite{wattenhofer2004distributed}, where a $5$-approximation factor is achieved w.h.p for general graphs, in $O(\log^2 n)$ rounds.
In \cite{lotker2009distributed} a randomized $(4+\epsilon)$-approximation for the weighted case is given, running in $O(\log n)$ rounds for constant $\epsilon>0$. This was later improved in \cite{lotker2008improved} to achieve a $(2+\epsilon)$-approximation in $O(\log \epsilon^{-1} \log n)$ rounds. In \cite{even2015distributed} a deterministic $(1+\epsilon)$-approximation is given, in $\Delta^{O(1/\epsilon)} + O(1/(\epsilon^2)) \cdot \log^* n$ rounds for the unweighted case, and $\log(\min \{ 1/ w_{min} , n / \epsilon \})^{O(1 / \epsilon)} \cdot (\Delta^{O(1/\epsilon)} + \log^* n)$ rounds for the weighted case, where the edge weights are in $[w_{min}, 1]$. In \cite{CzygrinowH03} a deterministic $(1+\epsilon)$-approximation is given, which finishes after $O(\log^{D(1/\eps)} n)$ rounds, where $D(1/\eps)$ is some function of $1/\eps$.
Due to \cite{kuhn2006price}, every algorithm achieving a constant approximation to the maximum matching problem requires $\Omega(\min\{\log \Delta/ \log\log \Delta, \sqrt{\log n / \log \log n}\})$ rounds.

The first distributed algorithm that uses the local ratio technique is due to~\cite{patt2008distributed}. The local ratio technique was also used in \cite{bar2016distributed} to compute a distributed $(2+\epsilon)$-approximation for weighted vertex cover. In \cite{panconesi2008fast}, a similar technique of weight grouping is used in the primal-dual framework for scheduling.

\section{MaxIS approximation}
\label{sec:MaxIS}
We begin, in Subsection~\ref{sec:sequentialMaxIS}, by showing the idea behind the use of local ratio for approximating MaxIS. This is done by presenting a sequential meta-algorithm and analyzing its correctness. Then, in Subsection~\ref{sec:distributedMaxIS}, we show how to implement this algorithm in the $\mathsf{CONGEST}$ model, and prove the claimed round complexity.

\subsection{Sequential MaxIS approximation via local ratio}
\label{sec:sequentialMaxIS}

Here we provide a sequential $\Delta$-approximation meta-algorithm to be used as the base for our distributed algorithm. The correctness of the algorithm is proved using the local ratio technique for maximization problems \cite{bar2001unified}.
We assume a given weighted graph $G=(V,w,E)$, where $w : V \rightarrow \mathbb{R}_+$ is an assignment of weights for the nodes and the degree of each node is bounded by $\Delta$. A simple $\Delta$-approximation local ratio algorithm exists for the problem \cite{bar2004local}. We rely on the following local ratio theorem for maximization problems \cite[Theorem 9]{bar2004local} in our proof.

\begin{theorem}
	Let $C$ be a set of feasibility constraints on vectors in $\mathbb{R}^n$. Let $w,w_1,w_2 \in \mathbb{R}^n$ be vectors such that $w=w_1+w_2$. Let $x\in \mathbb{R}^n$ be a feasible solution (with respect to $C$) that is $r$-approximate with respect to $w_1$ and with respect to $w_2$. Then $x$ is $r$-approximate with respect to $w$ as well.
\end{theorem}
In our case the vector $w$ is the weight vector representing the weight function of $G(V,w,E)$, $x$ is a binary vector indicating which nodes are chosen to the solution and the set of constraints $C$, is the set of independence constraints. We call the graph with weight vector $w_1$ the \emph{reduced graph}  and the graph with weight vector $w_2$ the \emph{residual graph}.

As standard practice with the local ratio technique, the splitting of the weight vector into $w_1,w_2$ is done such that any $r$-approximate solution to the reduced graph can be easily transformed into an $r$-approximate solution to the residual graph, while keeping it an $r$-approximate solution for the reduced graph. This allows us to apply weight reductions iteratively, solving each subproblem while maintaining the constraints. It is important to note that the theorem also holds if the weights in the reduced graph take negative values.

For the specific problem of MaxIS, we note that picking some node $v\in V$ and reducing the weight of $v$ from every $u \in N(v)$ splits the weight vector $w$ into two vectors, $w_1$ and $w_2$. Where $w_2(v)=w(v)$ for every $u\in N(v)$ and zero for every other node, and $w_1=w-w_2$.
Note that any $\Delta$-approximate solution for the reduced graph can be easily turned into a $\Delta$-approximate solution for the residual graph. This is done by making sure that at least some $u \in N(v)$ is in the solution: If this is not the case, we can always add one $u\in N(v)$ to the solution without violating the independence constraints. This only increases the value of our solution, making it $\Delta$-approximate for both the residual and the reduced graphs.

The above solution is sequential by nature. Implementing it directly in the distributed setting will require $O(n)$ rounds. We notice that if two nodes are in different neighborhoods of the graph then this process can be performed by both of them simultaneously without affecting each other. This observation forms the base for our distributed implementation.

We expand this idea by taking any independent set $U \subseteq V$ and for every $v \in U$ reducing the weight of $v$ from every $u \in N(v)$ in parallel. Next, solve the problem for the reduced graph. If for some $v \in U$, every $u\in N(v)$ is not in the solution for the reduced graph, we add $v$ to the solution for the reduced graph. This yields a $\Delta$-approximate solution for the problem.
For the sake of simplicity let $V=[n]$. Let $w_2$ be the weight vector of the residual graph after performing weight reductions as described above for some independent set $U \subseteq V$. By definition $w_2[v] = \sum_{u \in U \cap N(v)} w[u]$. The weight of the reduced graph is given by $w_1 = w - w_2$. Let $x \in \{0,1\}^n$ be some $\Delta$-approximate solution for the reduced graph. The cost of the solution $x$ is $\sum_{v}w[v]x[v]$. Let $x'\in \{0,1\}^n$ be defined as follows:
\begin{align}
	x'[u] =
	\begin{cases}
	1 & u\in U \wedge \forall v\in N(u), x[v]=0 \\
	x[u] & \text{otherwise}
	\end{cases}
\end{align}
 We prove the following lemma (See appendix A.1).
\begin{lemma}
	\label{lem:ind-set-lr}
	$x'$ is a $\Delta$-approximate solution for both the reduced graph and the residual graph.
\end{lemma}
%
\paragraph{Overview of Algorithm~\ref{alg:seq-maxis}:} The pseudocode is give in Algorithm~\ref{alg:seq-maxis} in appendix A. Using Lemma~\ref{lem:ind-set-lr} we construct a meta-algorithm that at each iteration picks an independent set $U \subseteq V$, reduces the weights of the elements in $U$ from their neighborhood and calls itself recursively with the reduced weights. This implicitly splits the graph into the reduced graph and the residual graph. A recursive call returns a $\Delta$-approximate solution for the reduced graph which is turned into a $\Delta$-approximate solution for both graphs by adding all nodes in the independent set that do not have neighbors in the returned solution. According to the local ratio theorem the final solution is a $\Delta$-approximation. Currently we are only interested in the correctness of the algorithm, thus it does not matter how the set $U$ is picked.
\RestyleAlgo{boxruled}
\LinesNumbered
\newcommand{\SeqAlgMaxIS}
{
\begin{algorithm}[htbp]
	\label{alg:seq-maxis}
	\caption{\texttt{SeqLR}($V,E,w$) - Sequential LR algorithm for maximum independent set}
	\footnotesize
	\If{$V=\emptyset$}{Return $\emptyset$}
	\ForEach{$v \in V$}{
	\If{$w(v) \leq 0 $}
	{
		$V = V \setminus \{v \}$\\
		$E = E \setminus \{(v,u) \mid u \in V \}$	
	}
	}
	Let $U \subseteq V$ be an independent set\\
	Let $w_1=w$ \\
	\ForEach{$u \in U$}{
		\ForEach{$v \in N(u)$}
		{
			$w_1(v) = w(v) - w(u)$	\\
		}		
	}
	$R =$ \texttt{SeqLR}($V,E,w_1$)\\
	$U = U \setminus \bigcup_{v \in R} N(v)$
	
	Return $R \cup U$\\
\end{algorithm}
}
The recursive step of Algorithm~\ref{alg:seq-maxis} returns a $\Delta$-approximate solution for the reduced graph which is then turned into a $\Delta$-approximate solution for the residual graph. Correctness follows from Lemma~\ref{lem:ind-set-lr} combined with a simple inductive argument. In the next section we implement this algorithm in a distributed setting.

\subsection{Distributed MaxIS approximation via local ratio}
\label{sec:distributedMaxIS}

In this section we implement Algorithm~\ref{alg:seq-maxis} in the distributed setting. We present an algorithm which iteratively finds independent sets and finishes after $\log \Wmax$ iterations. This yields a $\Delta$-approximation in $O(\MIS \log \Wmax)$ rounds, where $\MIS$ is the running time of a black-box MIS algorithm used. The algorithm that wraps the MIS procedure is deterministic, while the MIS procedure may be random. If the MIS procedure is random and finishes after $T$ rounds w.h.p then our algorithm requires $O(T\log \Wmax)$ rounds w.h.p. This holds for the $\mathsf{CONGEST}$ model. 

From now on we assume that all node weights are integers in $[\Wmax]$.
The sequential meta algorithm can be implemented distributedly, by having each node in the set perform weight reductions independently of other nodes. The key questions left open in the transition to the distributed setting is how to select our independent set at each iteration and how many rounds we need. Iteratively running the MIS procedure and performing weight reductions does not guarantee anything with regard to the number of nodes removed at each iteration or to the amount of weight reduced.

\newcommand{\DistAlgMaxIS}
{
\begin{algorithm}[htbp]
	\label{alg:dist-maxis}
	\caption{A distributed $\Delta$-approximation for weighted MaxIS, code for node $v$}
	
	//$w(v)$ is the initial weight of $v$, $w_v(v)$ changes during the run of the algorithm\\
	$w_v(v) = w(v) $\\
	$\ell_v(v) = \lceil \log w_v \rceil$\\
	$status = waiting$\\
	\While{true}
	{
		\ForEach{$reduce(x)$ received from $u\in N(v)$}
		{
			$w_v(v) = w_v(v) - x$\\
			$N(v) = N(v) \setminus \{ u \}$\\
			\If{$w_v(v) \leq 0$}
			{
				Send $\removed(v)$ to all neighbors\\
				return \NotInIS \\					
			}			
		}
		\ForEach{$\removed(u)$ received from $N(v)$}
		{
			$N(v) = N(v) \setminus \{ u \}$\\
		}
		$\ell_v(v) = \lceil \log w_v(v) \rceil$\\
		Send $\weightUpdate(v, w_v(v))$ to all neighbors\\		
		
		\ForEach{$\weightUpdate(u, w')$ received from $N(v)$}
		{
			$w_v(u) = w'$\\
			$\ell_v(u) = \lceil \log w_v(u) \rceil$\\
		}
		\If{$status = waiting$}{			
			\If{$\forall u\in N(v), \ell_v(u) \leq \ell_v(v)$}
			{
				$status(v) = ready$\\
				
				\While{$\exists u\in N(v), \ell_v(u) = \ell_v(v) $ and $status(u) \neq ready$}
				{
					\Wait
				}
				$v$ starts running MIS algorithm\\
			}
			
			\If{$v$ in MIS}
			{
				Send $reduce(w_v(v))$ to all neighbors\\
				$w_v(v) = 0$\\
				$status = candidate$
			}
			\Else
			{
				$status = waiting$
			}
	}
	\ElseIf{$status=candidate$}
	{
		\If{$N(v)=\emptyset$}
		{
			Send $\addedToIS(v)$ to all neighbors\\
			Return \InIS \\			
		}
		\If{$\addedToIS(u)$ received from $N(v)$}
		{
			Send $\removed(v)$ to all neighbors\\
			Return \NotInIS \\		
		}
	}
		
	}
\end{algorithm}
}

\paragraph{Overview of the distributed algorithm.} The pseudocode is give in Algorithm~\ref{alg:dist-maxis} in appendix A. The algorithm works by dividing the nodes into \emph{layers} according to their weights. The $i$-th layer is given by $L_i = \{v \mid 2^{i-1}< w(v) \leq 2^i \}$. During the algorithm each node keeps track of the weights (and layers) of neighboring nodes and updates are sent regarding weight changes and node removals. We divide the algorithm into two stages: the \emph{removal stage} and the \emph{addition stage}.

In the removal stage we find an independent set in the graph and perform weight reductions exactly as in the sequential meta algorithm. When finding the MIS, nodes in higher layers are prioritized over nodes in lower layers. A node cannot start running the MIS algorithm as long as it has a neighbor in a higher level. The most important thing to note here is that nodes in the topmost level never need to wait. A node who is selected to the MIS during the removal stage is a \emph{candidate node}. A node whose weight becomes zero or negative without being added to the MIS is said to be a \emph{removed node}. Removed nodes output \NotInIS~and finish, while candidate nodes continue to the addition stage. Both candidate and removed nodes are deleted from the neighborhood of their neighbors.
\par
In the addition stage, a candidate node $v$ remains only with neighbors with higher weights. We say these nodes have \emph{precedence} over the node $v$. A node $v$ may add itself to the solution only if it has no neighboring nodes which have precedence over it. After a node is added to the solution, all of its neighbors become removed. This corresponds line 13 in the sequential meta algorithm.

The correctness of the distributed algorithm follows directly from the correctness of the sequential meta algorithm. We are only left to bound the number of rounds. Let us consider the communication cost of the removal stage. We define the topmost layer to be $\Ltop=L_j$ where $j=arg\,max_{i} \, L_i \neq \emptyset$.
Note that nodes in $\Ltop$ never wait to run the MIS, and that after the MIS finishes for $\Ltop$, the weight of every $v \in \Ltop$ is reduced by at least a factor of two, emptying that layer. This can repeat at most $\log \Wmax$ times. \par
We assume a black-box MIS algorithm that finishes after $\MIS$ rounds with probability at least $1-p$. 
We now arrive at the main theorem for this section (See appendix A.1 for proof).

\begin{theorem}
	\label{thm:maxis-main}
	The distributed MaxIS approximation algorithm (Algorithm~\ref{alg:dist-maxis} in appendix A) finishes within $O(\MIS\cdot \log \Wmax)$ rounds with probability at least $1-p\log \Wmax$ in the $\mathsf{CONGEST}$ model.
	\footnote{The MIS algorithm is always executed on the entire graph $G$. Thus, its success probability does not change as we move between levels.}
\end{theorem}
%


\subsection{Deterministic coloring-based approximation algorithm}
\label{sec:deterministicMaxIS}
In this section we present a simple coloring-based $\Delta$-approximation algorithm for MaxIS. The advantage of this approach is that we have no dependence on $\Wmax$, yielding a deterministic algorithm running in $O(\Delta + \log^* n)$ rounds in the $\mathsf{CONGEST}$ model.

In the algorithm (pseudocode in Algorithm~\ref{alg:col-dist-maxis} in appendix A), instead of partitioning the nodes based on weights, they are partitioned based on colors, where colors with larger index have priority. Nodes perform weight reductions if their color is a local maxima. As in the previous section we have two stages: \emph{removal} and \emph{addition}, and three types of node states: \emph{removed}, \emph{candidate} and \emph{precedent}. After one iteration all nodes of the top color are either candidate or removed nodes. Thus after $\Delta+1$ iterations all nodes are either candidate or removed nodes. Thus, the removal stage finishes in $O(\Delta)$ rounds.

As in Algorithm~\ref{alg:dist-maxis}, after the removal stage all candidate nodes only have nodes who have precedence over them as their neighbors. A node adds itself to the independent set if it has no neighbors with precedence over it, in which case all of its neighbors become removed. We again note that candidate nodes of the smallest color have no neighbors and are added to the solution. Thus, the removal stages finishes in $O(\Delta)$ rounds.

\newcommand{\DistAlgColor}
{
\begin{algorithm}[htbp]
	\label{alg:col-dist-maxis}
	\caption{Coloring-based distributed $\Delta$-approximation for weighted MaxIS, code for node $v$}
	Run a $\Delta+1$ coloring algorithm\\
	Let $c:v \rightarrow [\Delta+1]$ be a coloring for the nodes\\
	$w(v) = w_v $\\
	\ForEach{$reduce(w')$ received from $u\in N(v)$}
	{
		$w(v) = w(v) - w'$\\
		$N(v) = N(v) \setminus \{ u \}$\\
		\If{$w(v) \leq 0$}
		{
			Send $removed(v)$ to all neighbors\\
			return \NotInIS \\					
		}			
	}
	\ForEach{$removed(u)$ received from $N(v)$}
	{
		$N(v) = N(v) \setminus \{ u \}$\\
	}		
	\If{$N(v)=\emptyset$}
	{
		Send $addedToIS(v)$ to all neighbors\\
		Return \InIS \\			
	}
	\If{$addedToIS(u)$ received from $N(v)$}
	{
		Send $removed(v)$ to all neighbors\\
		Return \NotInIS \\		
	}
	\If{$\forall u\in N(v)\setminus \{v\}$ it holds that $c(v) > c(u)$}
	{
		\ForEach{$u \in N(v)$}
		{
			Send $reduce(w(v))$
		}
		$w(v) = 0$\\
	}
\end{algorithm}
}
Algorithm~\ref{alg:col-dist-maxis} is a distributed implementation of Algorithm~\ref{alg:seq-maxis}, where the independent set is selected via its color at each iteration. The correctness of the algorithm follows from the correctness of Algorithm~\ref{alg:seq-maxis}. The number of rounds of Algorithm~\ref{alg:col-dist-maxis} is $O(\Delta + \log^* n)$ by using a deterministic distributed coloring algorithm of $O(\Delta + \log^* n)$ rounds \cite{barenboim2009distributed,Barenboim15}.\footnote{\cite{fraigniaud2015local} gives a faster coloring but is in $\mathsf{LOCAL}$ and in any case we need to pay for the number of colors as well.} The $\log^* n$ factor cannot be improved upon due to a lower bound by Linial \cite{Linial1987}. 

\subsection{Distributed $2$-approximation for maximum weighted matching}
\label{sec:2MM}
From the results in the previous section we can now derive local $2$-approximation algorithms for maximum matching. Let $G$ be a graph with weighted nodes, and let $L(G)$ be the line graph of $G$. 
It is well known that a maximum independent set in $L(G)$ corresponds to a maximum matching in $G$. An algorithm is executed on the line graph by assigning each edge in $G$ to have its computation simulated by one of its endpoints \cite{kuhn2005price}. We show that running our local ratio based approximation algorithms on $L(G)$ yields a $2$-approximate maximum matching in $G$. The main challenge is how to handle congestion, since nodes in $G$ may need to simulate many edges, thus may have to send many messages in a naive simulation \footnote{What follows is equivalent to iteratively running a maximal matching on weight groups in $G$ and performing local ratio steps on the edges of the matching. We go to $L(G)$ in order to demonstrate how a wide class of algorithms can be executed on the line graph while avoiding congestion.}.

Recall Algorithm~\ref{alg:seq-maxis}, the sequential $\Delta$-approximation meta-algorithm. The approximation factor was proved in Lemma~\ref{lem:ind-set-lr} to be $\Delta$. Specifically, the following equation provided an upper bound for the weight of an optimal solution $x^*$.
	$$\sum_{u\in U}\sum_{v \in N(u)} w[u]x^*[u] \leq \sum_{u\in U} w[u]\cdot (|N(u)|-1) = \sum_{u\in U} w[u]\cdot deg(u) \leq \Delta\sum_{u\in U} w[u].$$
The above bound uses the fact that for any node $u\in U$, at most $|N(u)|-1$ nodes in $N(v)$ can be selected for the solution due to independence constraints. But in $L(G)$ the largest independent set in the neighborhood of some node in $L(G)$ is at most 2, yielding the following upper bound:
	$\sum_{u\in U}\sum_{v \in N(u)} w[u]x^*[u] \leq \sum_{u\in U} 2w[u].$
We conclude that the algorithms presented in the previous sections provide $2$-approximation for maximum matching when executed on $G(L)$.

As for the communication complexity, the line graph has at most $n\Delta$ nodes and degree bounded by $2\Delta-2$. Thus, simulating our algorithms on $L(G)$ in the $\mathsf{LOCAL}$ model does not incur any additional asymptotic cost. However, in a naive implementation in the $\mathsf{CONGEST}$ model, we pay an $O(\Delta)$ multiplicative penalty due to congestion. This can be avoided with some modifications to our algorithms, as explained next.

For $e=(v,u)$ that is simulated by $v$, we call $v$ its \emph{primary} node and $u$ its \emph{secondary} node. We define a family of algorithms called \emph{local aggregation algorithms} and show that these algorithms can be augmented to not incur any additional communication penalty when executed on the line graph relative to their performance on the original graph in the $\mathsf{CONGEST}$ model. We begin with some definitions.


%

\begin{definition}
	We say that $f:\Sigma^n \rightarrow \Sigma$ is \emph{order invariant}, if for any set of inputs $\{x_i\}^n_{i=1}$, and any permutation $\pi$, it holds that $f(x_1,..., x_n) = f(x_{\pi(1)},..., x_{\pi(n)})$.
\end{definition}
For the sake of simplicity, if $f$ is order invariant we write $f(x_1,..., x_n)$ as $f(\{x_i\})$.
We may also give a partial parameter set to our function, in which case we assume all remaining inputs to the function are the empty character $\epsilon \in \Sigma$. Formally, for $X'=\{x_i\}^k_{i=1}$, denote $f(X')=f(x_1,...,x_k,\epsilon,...,\epsilon)$.
\begin{definition}
	We say that a function $f:\Sigma^n \rightarrow \Sigma$ is an \emph{aggregate function} if it is order invariant and there exists a function $\phi:\Sigma^2 \rightarrow \Sigma$ such that for any set of inputs $X=\{x_i\}^k$, and any disjoint partition of the inputs into $X_1,X_2$ it holds that $f(X)= \phi(f(X_1), f(X_2))$. The function $\phi$ is called the \emph{joining function}.
\end{definition}

\begin{observation}
	\label{col:bool-dec}
	It is easy to see that Boolean "and" and "or" functions are aggregate functions.
\end{observation}

Let $Alg$ be some algorithm for the $\mathsf{CONGEST}$ model. Let $\Dvi$ be the local data stored by $v$ during the round $i$ of the algorithm at. Let $\Dnvi= \{\Dui \mid u\in N(v) \}$ be the data of $v$'s immediate neighborhood.
\begin{definition}
	We call $Alg$ a \emph{local aggregation algorithm} if it only accesses $\Dnvi$ using aggregate functions where $|\Sigma| = O(\log n)$ and $|\Dvi|=O(\log n)$ for every $v \in V, i \in [t]$.
\end{definition}

We prove the following theorems (see appendix A.1):
\begin{theorem}
	\label{thm:aggregation-alg}
	If $Alg$ is a local aggregation algorithm running in the $\mathsf{CONGEST}$ model in $O(t)$ rounds, it can be executed on the line graph in $O(t)$ rounds.
\end{theorem}
%

\begin{theorem}
	\label{thm:dist-maxis-is-agg-alg}
	Algorithm~\ref{alg:dist-maxis} is a local aggregation algorithm.
\end{theorem}

This exact same technique can be applied to Algorithm~\ref{alg:col-dist-maxis}, giving the main result for this section:
\begin{theorem}
	There exist a randomized $2$-approximation algorithm for maximum weighted matching in the $\mathsf{CONGEST}$ model running in $O(\MIS\cdot\log \Wmax)$ rounds, and a deterministic $2$-approximation algorithm for maximum weighted matching in the $\mathsf{CONGEST}$ model running in $O(\Delta + \log^* n)$ rounds.
\end{theorem}

\section{Time-Optimal Approximations of Maximum Matching}
Here, we provide a sketch of our $O(\frac{\log \Delta}{\log\log \Delta})$-round algorithms for $(2+\eps)$-approximation of maximum weighted matching and $(1+\eps)$-approximation of maximum unweighted matching.   As stated before, these are the first to obtain the provably optimal round complexity, matching the $\Omega(\frac{\log \Delta}{\log\log \Delta})$ lower bound of \cite{kuhn2006price}, which holds for \emph{any} constant approximation. Full details appear in \Cref{sec:2epsMM}. 

\subsection{A fast $(2+\eps)$-approximation of maximum weighted matching}\label{lem:NMISmain}

We present here an $O(\frac{\log \Delta}{\log\log \Delta})$-round $(2+\eps)$-approximation for maximum unweighted matching. The extension to the weighted case follows known methods, and is explained in \Cref{lem:NMIS}. 

We develop a faster variant of an algorithm of Ghaffari~\cite{Ghaffari2016} that computes a nearly-maximal independent set. Our algorithm improves the $O(\log \Delta)$ round complexity of the nearly-maximal independent set algorithm of~\cite{Ghaffari2016} to $O(\frac{\log \Delta}{\log\log \Delta})$, which is optimal. To compute a $(2+\eps)$-approximation of maximum cardinality matching, we run this nearly-maximal independent set algorithm on the line graph of the network. Hence, it computes a nearly-maximal matching, which we show to be a $(2+\eps)$-approximation of maximum matching.

\paragraph{The Improved Nearly-Maximal Independent Set Algorithm}
In each iteration $t$, each node $v$ has a probability $p_t(v)$ for trying to join the independent set $\mathsf{IS}$. Initially $p_0(v)=1/K$, for $K=\Theta(\log^{0.1} \Delta)$. The total sum of the probabilities of neighbors of $v$ is called its \emph{effective-degree} $d_{t}(v)$, i.e., $d_t(v)=\sum_{u \in N(v)} p_{t}(u)$. The probabilities change over time as follows: $$p_{t+1}(v)=
\begin{cases}
    p_{t}(v)/K, & \text{if } d_{t}(v)\geq 2\\
    \min\{Kp_{t}(v), 1/K\},  &\text{if } d_{t}(v)< 2.
\end{cases}
$$
The probabilities are used as follows: In each iteration, node $v$ gets \emph{marked} with probability $p_{t}(v)$ and if no neighbor of $v$ is marked, $v$ joins $\mathsf{IS}$ and gets removed along with its neighbors.

The near-maximality of the computed independent set is captured by the following theorem, which shows that if we run the algorithm for $O(\log \Delta/\log K + K^2 \log 1/\delta)$ rounds, each node has only a $\delta$ probability not to be in the neighborhood of the computed independent set. It will be helpful to think of $\delta>0$ as a desirably small constant, but the bound grows quite slowly with $\log 1/\delta$. The proof of \Cref{thm:local-restate} is deferred to \Cref{lem:NMIS}. 
\begin{theorem} \label{thm:local-restate} For each node $v$, the probability that by the end of round $\beta(\log \Delta/\log K + K^2 \log 1/\delta)$, for a large enough constant $\beta$, node $v$ is not in $\mathsf{IS}$ and does not have a neighbor in $\mathsf{IS}$ is at most $\delta$. Furthermore, this holds even if coin tosses outside $N^{+}_{2}(v)$ are determined adversarially.
\end{theorem}
As a corollary of the nearly-maximal independent set algorithm, we get a $(2+\eps)$-approximation of maximum unweighted matching, by running it on the line graph of our network. The intuitive reason for this approximation is that, as \Cref{thm:local-restate} suggests, in expectation only a $\delta\ll \eps$ fraction of the maximum matching will not be in the neighborhood of the computed nearly-maximal matching. The formal claim is presented in the following theorem, the proof of which appears in \Cref{lem:NMIS}.
\begin{theorem} \label{thm:MCM2CONGEST} There is an algorithm in the $\mathsf{CONGEST}$ model that computes a $(2+\eps)$-approximation of maximum cardinality matching in $O(\frac{\log\Delta} {\log\log \Delta})$ rounds, for any constant $\eps>0$, w.h.p.
\end{theorem}

\subsection{A fast $(1+\eps)$-approximation of maximum cardinality matching}
We now discuss our $O(\frac{\log \Delta}{\log\log \Delta})$-round algorithm for $(1+\eps)$-approximation of maximum unweighted matching. Since the algorithm and its analysis are somewhat lengthy and technical, we can present only a bird's-eye view of them. The actual description is deferred to \Cref{sec:1epsMM}. We first discuss the algorithm for the $\mathsf{LOCAL}$ model, and then briefly discuss some of the ideas we use for extending it to the $\mathsf{CONGEST}$ model.

We follow a classical approach of Hopcroft and Karp\cite{HopcroftKarp1973}. Given a matching $M$, an augmenting path $P$ with respect to $M$ is a path that starts with an unmatched vertex, and alternates between non-matching and matching edges, and ends in an unmatched vertex. Flipping an augmenting path $P$ means removing $P \cap M$ edges from $M$ and replacing them with edges of $P\setminus M$. The approximation algorithm based on the method of Hopcroft and Karp\cite{HopcroftKarp1973} works as follows: For each $\ell =1$ to $O(1/\eps)$, find a \emph{maximal set of vertex-disjoint augmenting paths of length $\ell$}, and flip all of them.

\noindent The analysis of \cite{HopcroftKarp1973}  shows that this finds a $(1+\eps)$-approximation of maximum matching. Hence, all that we need to do is to find a maximal set of vertex-disjoint augmenting paths of length $\ell$. This problem can be formulated as a maximal independent set problem in a virtual graph, called the \emph{conflict graph}: put one vertex for each augmenting path of length $\ell$, and connect two vertices if their corresponding paths intersect. Each communication round on the conflict graph can be simulated in $O(\ell) = O(1)$ rounds in the network in the $\mathsf{LOCAL}$ model. Thus, if we could find an MIS in $O(\log \Delta/\log\log \Delta)$ rounds, we would be done. However, we only know how to compute a nearly-maximal independent set in this number of rounds.

When applied in this context, the guarantee that our nearly-maximal independent set algorithm provides is that each augmenting path of length $\ell$ has only a small $\delta$ probability of remaining (without any intersecting $\ell$-length augmenting path in the computed nearly-maximal set). However, this notion of near-maximality is not strong enough for us to be able to say that we still get a good approximation. For instance, one natural idea would be to simply discard the remaining augmenting paths (and thus also their vertices). However, this is not possible because with the current notion of near-maximality, each node may have a high probability of having at least one of the augmenting paths going through it remain. Notice that there are up to $\Delta^\ell$ such paths and we cannot afford to use a union bound over them.

In a nutshell, our approach is to provide a much tighter analysis of (a simple variation of) the algorithm, by leveraging the fact that the paths are short, having length $O(1/\eps)$. This tighter analysis allows us to say that if we run the algorithm for slightly more time, larger by an $O(1/\eps^2)$ factor, then the probability of each node having a remaining augmenting path will be small enough to allow us to discard such nodes. This tighter analysis is presented in a more general framework, which may be of independent interest. It concerns computing a nearly-maximal matching in a low-rank hypergraph, where each hyperedge contains a small number of vertices. The relation is that we can think of each augmenting path as one hyperedge, on the same set of vertices. These hyperedges would have rank $O(1/\eps)$, and a matching of hyperedges --- that is, a set of hyperedges that do not share a vertex --- would be the a set of vertex-disjoint paths. 

The above sketches our algorithm in the $\mathsf{LOCAL}$ model. Making this algorithm work in the $\mathsf{CONGEST}$ model brings in a range of new challenges. For instance, we cannot build the conflict graph explicitly and hence, the above algorithm does not work as is. We use a number of ideas, in order to extend the algorithm to the $\mathsf{CONGEST}$ model, which are described in \Cref{sec:1epsMMCONGEST}. 

One of the key ideas, which we find particularly interesting and may prove useful beyond our work, is a decentralized manner of performing the increase or decreases of the marking probabilities in the nearly-maximal independent set algorithm (as discussed in the previous subsection). Notice that now each augmenting path has a probability, which we would like to increase or decrease. Roughly speaking, we define an attenuation parameter $\alpha_{t}(v)$ for each node $v$ and we let the marking probability of each augmenting path be the multiplication of the attenuations of its vertices. Each node will decide on its own whether to increase or decrease its attentuation. Of course it is possible that some nodes of the path raise their attenuation and some lower it. However, we prove that in a long enough span of time and by choosing the increase or decrease parameters right, the net effect will still be in the correct direction, allowing us to mimic the analysis of the ideal $\mathsf{LOCAL}$ model algorithm, and thus prove the approximation guarantee.

\section{Discussion}
\label{sec:discussion}
This papers gives distributed approximation algorithms for maximum independent set and maximum matching in $\mathsf{CONGEST}$. We obtain a $\Delta$-approximation for the former using local-ratio techniques, and deduce a $2$-approximation for the latter by defining local aggregation algorithms and showing that they allow simulation on the line graph in the restricted $\mathsf{CONGEST}$ model, despite the need to simulate many nodes.

We then provide fast approximations that relax either the approximation factor to $2+\eps$ or the setting to an unweighted case, but have the optimal round complexity of $O(\log\Delta/\log\log\Delta)$.

One intriguing open question is whether this fast running time can also be obtained for the problem of finding a maximal independent set. Notice that by the algorithm we describe in \Cref{lem:NMISmain}, we can compute an almost maximal independent set in $O(\log\Delta/\log\log\Delta)$ rounds. Particularly, this is an independent set where the probability of each node remaining (without being, or having a neighbor, in the independent set) is at most $2^{-\log^{1-\gamma}\Delta}$, for any desirably small constant $\gamma>0$. However, to be able to extend this to a maximal independent set, we would need this failure probability to be at most $2^{-\Theta(\log\Delta)}$. Furthermore, given our algorithm and the lower bound of Kuhn et al.~\cite{kuhn2006price}, we now know that this is the complexity of finding a constant approximation for maximum matching. Similarly, by Bar-Yehuda et al.~\cite{bar2016distributed}, we also know that this is the complexity of finding a constant approximation for the vertex-cover problem. However, for the problem of finding a maximal independent set, there remains a $\log\log\Delta$ gap between the lower bound of~\cite{kuhn2006price} and the algorithm of Ghaffari~\cite{Ghaffari2016}.

\bibliographystyle{alpha}
\bibliography{MaximumMatching}
\appendix
\newpage
\section{Omitted Pseudocodes and Proofs}

Algorithm~\ref{alg:seq-maxis} is the pseudocode for the sequential local ratio MaxIS approximation.
\SeqAlgMaxIS
Algorithm~\ref{alg:dist-maxis} is the pseudocode for the distributed randomized MaxIS approximation.
\DistAlgMaxIS
Algorithm~\ref{alg:col-dist-maxis} is the pseudocode for the distributed deterministic MaxIS approximation.
\DistAlgColor

\subsection{Proofs omitted from section 2}
\begin{lemma-repeat}{lem:ind-set-lr}
	$x'$ is a $\Delta$-approximate solution for both the reduced graph and the residual graph.
\end{lemma-repeat}
\begin{proof}
	We note that $w_2[u]=w[u]$ for every $u \in U$. Thus, $w_1[u]=0$ for every $u \in U$. We do not incur any additional cost for the reduced graph because $x'$ is created by adding nodes from $U$ to $x$. Because  $x$ is $\Delta$-approximate for the reduced graph, so is $x'$.
	
	For the residual graph, only nodes in $\cup_{u \in U} N(u)$ have non zero weights.
	Let $x^*\in \{0,1\}^n$ be an optimal solution for the residual graph. We can bound from above the weight of $x^*$ by summing over the weights of $N(u)$ for every $u\in U$ where $x^*[u]=1$, taking into account that for any neighborhood $N(u)$, at most $|N(v)|-1$ nodes can be selected to a solution due to the independence constraints. We get the following upper bound for the weight $x^*$:
	\begin{align*}
		\sum_{v\in V} w_2[v]x^*[v] = \sum_{v\in V} \sum_{u\in U\cap N(v)} w_2[u]x^*[u] = \sum_{u\in U}\sum_{v \in N(u)} w_2[u]x^*[u]\\ \leq \sum_{u\in U} w_2[u]\cdot (|N(u)|-1) = \sum_{u\in U} w_2[u]\cdot deg(u) \leq \Delta\sum_{u\in U} w_2[u].
	\end{align*}
	On the other hand, $x'$ is selected such that for each $u\in U$ at least one $v \in N(u)$ is in $x'$ for any $u \in U$. Thus, $x'\cdot w_2 = \sum_{u\in U}\sum_{v \in N(u)} w_2[u]\cdot x'[u] \geq \sum_{u\in U} w_2[u]$, which means that $x'$ is at least a $\Delta$-approximation for $x^*$ on the residual graph, and the proof is complete.
\end{proof}
\begin{lemma}
	\label{lem:dist-mis-correctness}
	With probability at least $1-p$, $\Ltop = \emptyset$ after $O(\MIS)$ rounds.
\end{lemma}
\begin{proof}
	Let $G'$ be the graph induced by $\Ltop$. By the code, nodes in $\Ltop$ need not wait to run an MIS algorithm and as long as a node is not in $\Ltop$ it does not participate in an MIS algorithm. With probability at least $1-p$ an MIS is selected for $G'$ after $\MIS$ rounds.
	All nodes selected to the MIS have their weights reduced to zero. Every other node $v$ has at least one neighbor in the MIS, whose weight, by our definition of layers, is at least half of the weight of $v$. Thus the weight of every node $v\in \Ltop$ is halved, emptying the layer.
\end{proof}

\begin{theorem-repeat}{thm:maxis-main}
	The distributed MaxIS approximation algorithm (Algorithm~\ref{alg:dist-maxis}) finishes after at most $O(\MIS\cdot \log \Wmax)$ rounds with probability at least $1-p\log \Wmax$ in the $\mathsf{CONGEST}$ model. 
\end{theorem-repeat}
\begin{proof}
	Applying a union bound over all layers, gives that all layers are empty after at most $O({\MIS\cdot \log \Wmax})$ iterations with probability at least $1-p\log \Wmax$, by Lemma~\ref{lem:dist-mis-correctness}. We require $p=o(1/\log \Wmax)$. This bounds the communication cost for the removal stage.
	
	Denote by $C_i$ the set of candidate nodes from level $L_i$. These nodes are at level $L_i$ when they are set to be candidate nodes.
	Nodes in $C_i$ wait for neighbors with higher precedence to decided whether they enter the solution. We note that nodes in $C_0$ do not have any neighbors with higher precedence. After nodes in $C_0$ have decided, the nodes in $C_1$ do not have to wait and so on. Thus, all candidate nodes make a decision after at most $\log \Wmax$ rounds. This bounds the communication cost for the addition stage.
\end{proof}
\begin{theorem-repeat}{thm:aggregation-alg}
	If $Alg$ is a local aggregation algorithm running in the $\mathsf{CONGEST}$ model in $O(t)$ rounds, it can be executed on the line graph in $O(t)$ rounds.
\end{theorem-repeat}
\begin{proof}
	$Alg$ is executed on the primary node, and we maintain the invariant that $\Dvi$ is always present in both the primary and secondary nodes. Every time $Alg$ needs to execute a function $f$, both the primary and secondary nodes already have the data of all of their neighbors. Each node calculates $f$ on the data of its neighbors, the secondary node sends this calculation to the primary which in turn executes the joining function yielding the desired result. Afterwards the new node data is sent to the secondary node.
	
	No communication is needed to access the data of the neighbors, as a neighbor of $e$ must share a node with it, which contains its data. There is no congestion when sending the value of $f$ or the new data to the secondary node. Thus, the number of rounds is $O(t)$.
\end{proof}

\begin{theorem-repeat}{thm:dist-maxis-is-agg-alg}
	Algorithm~\ref{alg:dist-maxis} is a local aggregation algorithm.
\end{theorem-repeat}
\begin{proof}
	Let us explicitly define $D_{v,i}$ for every $v\in V$. Each node knows its weight, status and degree. Formally, $D_{v,i} = \{w_i(v), status_v, deg_i(v)\}$.
	The algorithm uses "and" and "or" Boolean functions, which by Observation~\ref{col:bool-dec}, are aggregate functions. Each node also needs to update its weight at each iteration. The weight update function for $v$ can be written as $f_w: [W]^{{v}\times N(v)} \rightarrow [W]$, with $f_w(w_v, \{w_u \mid u \in N(v) \}) = w_v - \sum w_u$, which is of course an aggregate function.	
\end{proof}
\section{Faster Approximations of Maximum Matching}
\label{sec:2epsMM}
In this section, we present $O(\frac{\log \Delta}{\log\log \Delta})$-round algorithms for $(2+\eps)$-approximation of maximum weighted matching and $(1+\eps)$-approximation of maximum unweighted matching. As stated before, these are the first algorithms to obtain the provably optimal round complexity for matching approximation. Their complexity matches the seminal lower bound of Kuhn, Moscibroda, and Wattenhofer~\cite{kuhn2006price} which shows that $\Omega(\frac{\log \Delta}{\log\log \Delta})$ rounds are necessary, in fact for \emph{any} constant approximation.

\subsection{A fast $(2+\eps)$-approximation of maximum weighted matching}\label{lem:NMIS}
We first present a simple $O(\frac{\log \Delta}{\log\log \Delta})$-round $(2+\eps)$-approximation for maximum unweighted matching. We then explain how this approximation extends to the weighted setting via known methods.

To get a $(2+\eps)$-approximation, we gradually find large matchings and remove them from along with the other edges that are incident on them. At the end, we show that the remaining graph has only a small matching left, hence allowing us to prove an approximation guarantee. 

The key algorithmic component in our approach is an adaptation of the algorithm of Ghaffari~\cite{Ghaffari2016}. Ghaffari presented an MIS algorithm, which if executed on a graph $H$ with maximum degree $\Delta$ for $O(\log \Delta + \log 1/\delta)$ rounds, computes an independent set $\mathsf{IS}$ of nodes of $H$, with the following probabilistic near-maximality guarantee: each node of $H$ is either in $\mathsf{IS}$ or has a neighbor in it, with probability at least $1-\delta$. We will be applying a similar method on the line graph of our original graph, hence choosing a nearly-maximal independent set of edges. However, this running time is not quite fast enough for our target complexity.

We first explain relatively simple changes in the algorithm and its analysis that improve the complexity to $O(\frac{\log \Delta}{\log\log \Delta})$, for any constant $\eps>0$. We then explain how that leads to an $O(\frac{\log \Delta}{\log\log \Delta})$-round $(2+\eps)$-approximation for maximum unweighted matching.\\

\begin{mdframed}[hidealllines=false,backgroundcolor=gray!30]
\paragraph{The Modified Nearly-Maximal Independent Set Algorithm}
In each iteration $t$, each node $v$ has a probability $p_t(v)$ for trying to join the independent set $\mathsf{IS}$. Initially $p_0(v)=1/K$, for a parameter $K$ to be fixed later. The total sum of the probabilities of neighbors of $v$ is called its \emph{effective-degree} $d_{t}(v)$, i.e., $d_t(v)=\sum_{u \in N(v)} p_{t}(u)$. The probabilities change over time as follows: $$p_{t+1}(v)=
\begin{cases}
    p_{t}(v)/K, & \text{if } d_{t}(v)\geq 2\\
    \min\{Kp_{t}(v), 1/K\},  &\text{if } d_{t}(v)< 2.
\end{cases}
$$
The probabilities are used as follows: In each iteration, node $v$ gets \emph{marked} with probability $p_{t}(v)$ and if no neighbor of $v$ is marked, $v$ joins $\mathsf{IS}$ and gets removed along with its neighbors.
\end{mdframed}

\begin{theorem-repeat}{thm:local-restate} 
For each node $v$, the probability that by the end of round $\beta(\log \Delta/\log K + K^2 \log 1/\delta)$ $=O(\frac{\log \Delta}{\log\log \Delta})$, for a large enough constant $\beta$, node $v$ is not in $\mathsf{IS}$ and does not have a neighbor in $\mathsf{IS}$ is at most $\delta$. Furthermore, this holds even if coin tosses outside $N^{+}_{2}(v)$ are determined adversarially.
\end{theorem-repeat}
Let us say that a node $u$ is \emph{low-degree} if $d_t(u)<2$, and \emph{high-degree} otherwise. We define two types of \emph{golden rounds} for a node $v$: (1) rounds in which $d_t(v)<2$ and $p_{t}(v)= 1/K$, (2) rounds in which $d_{v}(t)\geq 1$ and at least $d_{t}(v)/(2K^2)$ of $d_{t}(v)$ is contributed by low-degree neighbors.

\begin{lemma}\label{lem:goldCount} By the end of round $\beta(\log \Delta/\log K + K^2 \log 1/\delta)$, either $v$ has joined $\mathsf{IS}$, or has a neighbor in $\mathsf{IS}$, or at least one of its golden round counts reached $\frac{\beta}{13}(\log \Delta/\log K + K^2 \log 1/\delta)$.
\end{lemma}
\begin{proof} Let $T=\beta(\log\Delta/\log K + K^2 \log 1/\delta)$ for a sufficiently large constant $\beta$. We focus only on the first $T$ rounds. Let $g_1$ and $g_2$ respectively be the number of golden rounds of types 1 and 2 for $v$, during this period. We assume that by the end of round $T$, node $v$ is not removed and $g_1 \leq T/13$, and we conclude that $g_2 \geq T/13$.

Let $h$ be the number of rounds during which $d_{t}(v)\geq 2$. Notice that the changes in $p_{t}(v)$ are governed by the condition  $d_{t}(v)\geq 2$ and the rounds with $d_{t}(v)\geq 2$ are exactly the ones in which $p_{t}(v)$ decreases by a $K$ factor. Since the number of $K$-factor increases in $p_{t}(v)$ can be at most equal to its number of $K$-factor decreases, there are at least $T -2h$ rounds in which $p_{t}(v)=1/K$. Out of these rounds, at most $h$ rounds can have $d_{t}(v)\geq 2$. Hence, $g_1 \geq T -3h$. The assumption $g_1 \leq T/13$ gives that $h\geq 4T/13$. Let us now consider the changes in the effective-degree $d_{t}(v)$ of $v$ over time. If $d_{t}(v) \geq 1$ and this is not a golden round of type-2, then we have 
$$d_{t+1}(v) \leq K \frac{1}{2K^2} d_t(v)+ \frac{1}{K} (1-\frac{1}{2K^2}) d_{t}(v) < \frac{3}{2K} d_{t}(v).$$ 
There are $g_2$ golden rounds of type-2. Except for these, whenever $d_{t}(v)\geq 1$, the effective-degree $d_{t}(v)$ shrinks by at least a $\frac{3}{2K}$ factor. In these exception cases, it increases by at most a $K$ factor. Each of these exception rounds cancels the effect of no more than $2$ shrinkage rounds, as $(\frac{3}{2K})^2 \cdot K \ll 1$. Thus, ignoring the total of at most $3g_2$ rounds lost due to type-2 golden rounds and their cancellation effects, every other round with $d_{t}(v)\geq 2$ pushes the effective-degree down by a $\frac{3}{2K}$ factor. This cannot happen more than $\log_{\frac{2K}{3}}\Delta$ times as that would lead the effective degree to exit the $d_{t}(v)\geq 2$ region. Hence, the number of rounds in which $d_{t}(v)\geq 2$ is at most $\frac{\log\Delta}{\log \frac{2K}{3}} + 3g_2$. That is, $h \leq \frac{\log\Delta}{\log \frac{2K}{3}} + 3g_2$. Since $h\geq 4T/13$, and because $T=\beta(\log\Delta/\log K + K^2 \log 1/\delta)$ for a sufficiently large constant $\beta$, we get that $g_2 > T/13$.
\end{proof}

\begin{lemma}\label{lem:removalPerRound} In each type-1 golden round, with probability at least $\Theta(1/K)$, $v$ joins the IS. Moreover, in each type-2 golden round, with probability at least $\Theta(1/K^2)$, a neighbor of $v$ joins the IS.  Hence, the probability that by the end of round $\beta(\log\Delta/\log K + K^2 \log 1/\delta)$, node $v$ has not joined the IS and does not have a neighbor in it is at most $\delta$. These statements hold even if the coin tosses outside $N^+_{2}(v)$ are determined adversarially.
\end{lemma}
\begin{proof} In each golden type-$1$ round, we have $d_t(v)<2$ and $p_{t}(v)= 1/K$. The latter means that node $v$ gets marked with probability $1/K$, and the former means that the probability that none of the neighbors of $v$ is marked is at least $\prod_{u\in N_{t}(v)} (1-p_{t}(u)) \geq 4^{-\sum_{u\in N_{t}(v)} p_{t}(u)} = 4^{-d_{t}(v)}\geq 1/16$. Hence, in each golden type-$1$ round, node $v$ joins the IS with probability at least $1/(16K)$. 

In each golden type-$2$ rounds, we have $d_{v}(t)\geq 1$ and at least $d_{t}(v)/(2K^2)$ of $d_{t}(v)$ is contributed by low-degree neighbors. Suppose we examine the set $L_{t}(v)$ of low-degree neighbors of $v$ one by one and check whether they are marked or not. We stop when we reach the first marked node. The probability that we find at least one marked node is at least $1-\prod_{u\in L_{t}(v)} (1-p_{t}(u)) \geq 1-e^{-\sum_{u\in L_{t}(v)} p_{t}(u)} \geq 1-e^{-d_{t}(v)/(2K^2)} \geq 1-e^{-1/2K^2}\geq 1/4K^2$, given that $(2K^2)\geq 1$. Now that we have found the first marked light neighbor $u$, the probability that no neighbor $w$ of $u$ is marked is at least $\prod_{w\in N_{t}(u)} (1-p_{t}(w)) \geq 4^{-\sum_{w\in N_{t}(u)} p_{t}(w)} = 4^{-d_{t}(u)}\geq 1/16$. Therefore, overall, the probability that node $v$ gets removed in a type-2 golden round is at least $\frac{1}{64K^2}$. 

Now notice that these events are independent in different rounds. Hence, the probability that node $v$ does not get removed after $\Theta(K^2\log 1/\delta)$ golden rounds is at most $(1-\frac{1}{64K^2})^{\Theta(K^2\log 1/\delta)} \leq \delta$. Furthermore, in the above arguments, we only relied on the randomness in the nodes that are at most within $2$ hops of $v$. Hence, the guarantee is independent of the randomness outside the $2$-neighborhood of $v$.
\end{proof}
\begin{proof}[Proof of \Cref{thm:local-restate}]
By \Cref{lem:goldCount}, within the first $\beta(\log \Delta/\log K + K^2 \log 1/\delta)$ round, each node $v$ is either already removed (by joining or having a neighbor in the IS) or one of its golden round counts reaches at least $\beta(\log \Delta/\log K + K^2 \log 1/\delta)/13$. As \Cref{lem:removalPerRound} shows, in each golden round, node $v$ gets removed with probability at least $\Theta(1/K^2)$. Hence, given a large enough constant $\beta$, the probability that node $v$ remains through $\beta(K^2 \log 1/\delta)/13$ golden rounds is at most $\delta$.
\end{proof}

\begin{theorem-repeat}{thm:MCM2CONGEST} There is a distributed algorithm in the $\mathsf{CONGEST}$ model that computes a $(2+\eps)$-approximation of maximum unweighted matching in $O(\frac{\log\Delta} {\log\log \Delta})$ rounds, for any constant $\eps>0$, whp.
\end{theorem-repeat}
\begin{proof} The algorithm executes the nearly-maximal independent set algorithm explained above on the line-graph. This finds a nearly-maximal set of independent edges, i.e., edges which do not share an endpoint, or in other words, a nearly-maximal matching. The fact that the algorithm can be run on the line-graph in the $\mathsf{CONGEST}$ model follows from \Cref{sec:2MM}, since it is easy to see that this is a local aggregation algorithm. The round complexity of $O(\frac{\log\Delta} {\log\log \Delta})$ follows from the $O(\log \Delta/\log K + K^2 \log 1/\delta)$ bound of \Cref{thm:local-restate}, by setting $K=\Theta(\log^{0.1} \Delta)$ and $\delta = 2^{-\log^{0.7} \Delta}$. Let us now examine the approximation factor. Each edge of the optimal matching has probability at most $\delta$ of becoming unlucky and not being in our found matching and not having any adjacent edge in it either. These are the edges that remain after all iterations of the nearly-maximal independent set algorithm. Thus, we expect at most $\delta \ll \eps$ fraction of the edges of the optimal matching to become unlucky. The number also has an exponential concentration around this mean\footnote{This concentration is due to the fact that the dependencies are local and each edge's event of being unlucky depends on only at most $\Delta$ other edges. However, one can obtain a better success probability. See \Cref{sec:2epsMMbipartite} for an algorithm which provides a stronger concentration, giving a $(2+\eps)$-approximation with probability $1-e^{-\Omega(-|OPT|)}$.}. Ignoring these $\eps|OPT|$ unlucky edges of the optimal matching, among the rest of the edges, each edge of the found matching can be blamed for removing at most $2$ edges of the optimal matching. So the found matching is a $(2+\eps)$-approximation.
\end{proof}

\paragraph{Extension to the Weighted Case via Methods of Lotker et al.} Above, we explain an $O(\frac{\log \Delta}{\log\log \Delta})$-round algorithm for $(2+\eps)$-approximation of maximum unweighted matching. This can be extended to the weighted case via known methods, while preserving the asymptotic complexity, as follows: First, we sketch a method of Lotker et al.\cite{lotker2009distributed} which allows one to turn a $(2+\eps)$-approximation for the unweighted case to an $O(1)$-approximation for the weighted case. Classify the weights of edges into powers of a large constant $\beta$, i.e., by defining weight buckets of the form $[\beta^{i}, \beta^{i+1}]$. In each of these big-buckets, partition the weight range further into $O(\log_{1+\eps} \beta)$ \emph{small-buckets} in powers of $1+\eps$. Run the following procedure in all big-buckets in parallel: Starting from the edges of the highest weight small-bucket in this big-bucket, find a $(2+\eps)$-approximation of the matching in that small-bucket using the unweighted matching algorithm, remove all their incident edges in that big-bucket, and move to the next biggest small-bucket. After $O(\log_{1+\eps} \beta)$ iterations of going through all the small-buckets, for each big-bucket, we have found a matching that is a $2+O(\eps)$ 
approximation of the maximum weight matching among all the edges with weight in this big-bucket. However, altogether, this is not a matching as a node might have a ``matching''-edge incident on it in each of the big-buckets. Keep each of these chosen edges only if it has the highest weight among the chosen edges incident on it.  Lotker et al.\cite{lotker2009distributed} showed that this produces an $O(1)$-approximation of the maximum weight matching.

Now, this $O(1)$-approximation can be turned into a $(2+\eps)$-approximation. Lotker et al.~\cite[Section 4]{lotkerMatchingImproved} present a method that via $O(1/\eps)$ black-box usages of an $O(1)$-approximation Maximum Weight Matching algorithm $\mathcal{A}$, produces a $(2+\eps)$-approximation of the maximum weight matching. We here provide only a brief and intuitive sketch. The method is iterative, each iteration is as follows. Let $M$ be the current matching. We look only at weighted augmenting paths of $M$ with length at most $3$. We define an auxiliary weight for each unmatched edge $e$, which is equal to the overall weight-gain that would be obtained by adding $e$ to the matching and instead erasing the matching $M$ edges incident on endpoints of $e$ (if there are any). Note that this auxiliary weight can be computed easily in $O(1)$ rounds. Then, we use algorithm $\mathcal{A}$ to find a matching which has an auxiliary weight at most an $O(1)$ factor smaller than the maximum weight matching, according to the auxiliary weights. Then we augment $M$ with all these found matching edges, erasing the previously matching edges incident on their endpoints. We are then ready for the next iteration. As Lotker et al. show, after $O(1/\eps)$ iterations, the matching at hand is a $(2+\eps)$-approximation of the maximum weight matching.

\subsection{A fast $(1+\eps)$-approximation of maximum cardinality matching in $\mathsf{LOCAL}$}
\label{sec:1epsMM}
Here, we present an $O(\frac{\log \Delta}{\log\log \Delta})$-round algorithm in the $\mathsf{LOCAL}$ model for $(1+\eps)$-approximation of maximum unweighted matching for any constant $\eps>0$. In the next subsection, we explain how to extend a variant of this algorithm to the $\mathsf{CONGEST}$ model, in essentially the same round complexity, i.e., without incurring a loss in the asymptotic notation.

Our algorithm follows a general approach due to the classical work of Hopcroft and Karp\cite{HopcroftKarp1973}, where one iteratively augments the matching with short augmenting paths, until the desired approximation (or exact bound) is achieved. In our case, the key algorithmic piece is to efficiently find a nearly-maximal set of disjoint short augmenting paths, in $O(\log \Delta/\log\log \Delta)$ rounds. Our base will again be Ghaffari's MIS algorithm\cite{Ghaffari2016}. However, here we need significant changes to the algorithm and its analysis.

\paragraph{Augmenting Paths.} Consider an arbitrary matching $M$. A path $P={v_0, v_1, v_2, \dots, v_{p}}$ is called an \emph{augmenting path} of length $p$ for $M$ if in $M$, we have the following two properties: (1) nodes $v_0$ and $v_p$ are unmatched, and (2) for every $2i+1\in [1, p-1]$, node $v_{2i+1}$ is matched to node $v_{2i+2}$. In this case, let $M\oplus P=M\cup P \setminus (M\cap P)$. That is, $M\oplus P$ is the matching obtained by erasing the matching edges $\{v_{1}, v_{2}\}, \{v_{3}, v_{4}\}, \dots, \{v_{p-2}, v_{p-1}\}$, and instead adding edges $\{v_{0}, v_{1}\}, \{v_{2}, v_{3}\}, \dots, \{v_{p-1}, v_{p}\}$. Note that $M \oplus P$ is indeed a matching, and moreover, it has one more matching edge than $M$. The operation of replacing $M$ with $M\oplus P$ is called \emph{augmenting} the matching $M$ with the path $P$.

For a fast distributed algorithm, we would like to be able to augment the matching $M$ with many augmenting paths simultaneously. For a given matching $M$, two augmenting paths $P_1$ and $P_2$ are called \emph{dependent} if their node-sets intersect. Note that in that case, we can not augment $M$ simultaneously with both $P_1$ and $P_2$. However, in case we have two independent augmenting paths, we can augment $M$ with both of them simultaneously, and the result is the same as first performing the first augmentation, and then performing the second.

We now recall two well-known facts about augmenting paths, due to the classical work of Hopcroft and Karp\cite{HopcroftKarp1973}: (1) Matching $M$ is a $(1+\eps)$-approximation of the maximum matching if and only if it does not have an augmenting path of length at most $2\lceil{1/\eps}\rceil+1$. (2) If the shortest augmenting path for $M$ has length $\ell$ and one augments $M$ with a maximal independent set of augmenting paths of length $\ell$, the shortest augmenting path of the resulting matching will have length at least $\ell+1$.

\paragraph{General Methodology.} Based on the above two facts, a natural and by now standard method for computing a $(1+\eps)$-approximation of maximum matching is as follows: for each $\ell =1, \dots, 2\lceil{1/\eps}\rceil+1$, we find a maximal independent set of augmenting paths of length exactly $\ell$ and augment the matching with them. At the end, we have a $(1+\eps)$-approximation of maximum matching. This outline was followed by Fischer et al.\cite{fischer1993approximating} in the PRAM model and Lotker et al.\cite{lotkerMatchingImproved} in the distributed model.

As clear from the above outline, the core algorithmic piece is to compute a maximal independent set of augmenting paths of length $\ell$. We consider an auxiliary graph with one node per each augmenting path of length $\ell$ and an edge between each two of them if they intersect. This auxiliary graph, which is usually called the \emph{conflict graph}, can be constructed and simulated in $\ell = O(1/\eps)$ rounds of communication on the base graph in the $\mathsf{LOCAL}$ model. Thus, the remaining question is how to find a maximal independent set on this graph. Lotker et al.\cite{lotkerMatchingImproved} used a variant of Luby's distributed MIS algorithm\cite{luby1986simple} to compute this set in $O(\log n)$ rounds. However, aiming at the complexity of $O(\log \Delta/\log\log \Delta)$, we cannot afford to do that. Indeed, it remains open whether an MIS can be computed in $O(\log \Delta/\log\log \Delta)$ rounds. Thus, unless we resolve that question, we cannot compute a truly maximal independent set. Our remedy is to resort to computing ``nearly-maximal'' sets, using ideas similar to the algorithm of the previous section. Here, the near-maximality should be according to an appropriate definition which allows us to preserve the approximation guarantee. However, there are crucial subtleties and challenges in this point, which require significant alterations in the algorithm, as discussed next.

\paragraph{Intuitive Discussion of the Challenge.} To be able to follow the general method explained above and get its approximation guarantee, we need to ensure that no short augmenting path remains. However, the set of augmenting paths that we compute are not exactly maximal, which means some paths might remain. A natural solution would be that, after finding a ``nearly-maximal'' set of augmenting paths of a given length, we \emph{neutralize/deactivate} the rest, say by removing one node of each of these remaining paths from the graph. However, to ensure that we do not lose in the approximation factor, we need to be sure that this removal does not damage the maximum matching size significantly. For instance, if we can say that each node is removed with a small probability $\delta \ll \eps$, in expectation this can remove at most a $2\delta$ fraction of the optimal matching edges, and thus $(1+\eps')$-approximating the remaining matching would give an approximation of roughly $(1+\eps')(1+2\delta) \approx 1+\eps'+2\delta$. This is a good enough approximation, as we can choose the $\eps'$ and $\delta$ appropriately, e.g., about $\eps/10$. However, in our context, running the nearly-maximal independent set algorithm of the previous subsection among augmenting paths for $O(\log\Delta/\log K + K^2 \log 1/\delta)$ rounds would only guarantee that the probability of each one augmenting path remaining is at most a small $\delta$. Since there can be up to $\Delta^{O(1/\eps)}$ augmenting paths going through one node, and as we want the running time within $O(\frac{\log \Delta}{\log\log \Delta})$, we cannot afford to apply a union bound over all these paths and say that the probability of each one node having a remaining augmenting path is small. The fix relies on some small changes and a much tighter analysis of the nearly-maximal independent set algorithm for this special case, leveraging the fact we are dealing with paths of constant length at most $d= O(1/\eps)$. In fact, to present the fix in its general form, we turn to another (equivalent) formulation of finding nearly-maximal matchings---that is, a set of hyperedges where each node has at most one of its hyperedges in this set---in a hypergraph $H$ of rank $d=O(1/\eps)$. The connection is that we will think of augmenting paths as hyperedges of $H$ and nodes of $H$ will be the same as nodes of the original graph $G$, where a hyperedge in $H$ includes all nodes of the corresponding augmenting path in $G$.

\paragraph{Nearly-Maximal Matching in Low-Rank Hypergraphs.} We want an algorithm for hypergraphs of rank $d$ that in $O(d^2 \frac{\log \Delta}{\log\log \Delta})$ rounds, deactivates each node with probability at most $\delta \ll \eps$, and finds a maximal matching in the hypergraph induced by active nodes. Note that this is stronger than guaranteeing that each edge is removed or has an adjacent edge in the matching with such a probability. The algorithm will be essentially the same as that of the previous subsection, where now each hyperedge $e$ has a probability $p_{t}(e)$ for each iteration $t$ and gets marked and joins the matching accordingly. We will however deactivate some nodes in the course of the algorithm. The more important new aspect is in the analysis.

\paragraph{The Change in the Algorithm.} Call a hyperedge $e$ \emph{light} iff $\sum_{e', e'\cap e\neq\emptyset} p_{t}(e') < 2$, and let $L_t$ be the set of light hyperedges of round $t$. Set $K=\log^{0.1}\Delta$. Call a round $t$ \emph{good} for a node $v$ if $\sum_{e, e\in L_t, v\in e} p_{t}(e) \geq 1/(2dK^2)$. Note that in a round that is good for $v$, with probability at least $\Theta(\frac{1}{dK^2})$, one of these light hyperedges joins the matching and thus $v$ gets removed. Deactivate node $v$ if it has had more than $\Theta(dK^2 \log 1/\delta)$ good rounds. Note that the probability that a node $v$ survives through $\Theta(dK^3\log 1/\delta)$ good rounds and then gets deactivated is at most $\delta$.

\paragraph{Analysis.} A key property of the algorithm is the following deterministic guarantee, which proves the maximality of the found matching in the hypergraph induced by active nodes:

\begin{lemma}\label{lem:MaximalityAmongActives} After $T=O(d^2 \frac{\log \Delta}{\log\log \Delta})$ rounds, there is no hyperedge with all its nodes active.
\end{lemma}
\begin{proof} We consider an arbitrary hyperedge $e=\{v_1, v_2, \dots, v_d\}$ and prove that it cannot be the case that all of its nodes remain active for $O(d^2 \frac{\log \Delta}{\log\log \Delta})$ rounds. We emphasize that this is a deterministic guarantee, and it holds for every hyperedge $e$. We assume that hyperedge $e$ is not removed (by removing itself due to an adjacent edge in the matching, or because of one of its nodes becoming deactivated) in the first $T=O(d^2 \frac{\log \Delta}{\log\log \Delta})$ rounds and we show that this leads to a contradiction, assuming a large enough constant in the asymptotic notation definition of $T$.

For each node $v$, call a round \emph{heavy} if $\sum_{e, v\in e} p_{t}(e) \geq 1/d$. If round $t$ is heavy but not good, then by definition of good rounds, at most $1/(2dK^2)$ weight in the summation comes from light hyperedges. This is at most a $\frac{1/(2dK^2)}{1/d} =1/2K^2$ fraction of the summation. Hence, in every heavy but not-good round, the summation shrinks by a factor of $1/(2K^2) \cdot K + (1-1/(2K^2)) \cdot 1/K\leq 2/K$. In each heavy and good round, the summation grows by at most a $K$ factor, which is in effect like canceling at most $2$ of the shrinkage rounds (as in the proof of Lemma~\ref{lem:goldCount}). The number of good rounds is at most $\Theta(dK^2\log 1/\delta)$. Therefore, since $\sum_{e, v\in e} p_{t}(e)$ starts with a value of at most $\Delta^d/K$, node $v$ can have at most $\Theta(dK^2\log 1/\delta) + 3d\log_{K}{\Delta}$ heavy rounds.

Now, looking at a hyperedge hyperedge $e=\{v_1, v_2, \dots, v_d\}$, we claim that $e$ cannot have more than $d\big(\Theta(dK^2\log 1/\delta) + 3d\log_{K}{\Delta}\big)$ rounds in which $\sum_{e', e'\cap e\neq\emptyset} p_{t}(e') \geq 1$. This is because in every such round, the summation in at least one of the $d$ nodes constituting $e$ must be at least $1/d$. Thus, in every such round, at least one of the nodes of hyperedge $e$ is heavy. But we just argued that each node has at most $\Theta(dK^2\log 1/\delta) + 3d\log_{K}{\Delta}$ heavy rounds. Thus, in total edge $e$ cannot have more than $h=\Theta(d^2K^2\log 1/\delta) + 3d^2\log_{K}{\Delta}$ rounds in which $\sum_{e', e'\cap e\neq\emptyset} p_{t}(e') \geq 2$.

During each round in which $\sum_{e', e'\cap e\neq\emptyset} p_{t}(e') \geq 2$, hyperedge $e$ reduces its $p_{t}(e)$ by a $K$ factor. Each other round raises $p_{t}(e)$ by a $K$ factor, unless $p_{t}(e)$ is already equal to $1/K$. Since each $K$-factor raise cancels one $K$-factor shrinkage, and as $p_{t}(e)$ starts at $1/K$, with the exception of $2h$ rounds, all remaining rounds have $p_{t}(e)=1/K$. Among these, at most $h$ can be rounds in which $\sum_{e', e'\cap e\neq\emptyset} p_{t}(e') \geq 2$. Thus, hyperedge $e$ has $T-3h$ rounds in which $p_{t}(e)=1/K$ and $\sum_{e', e'\cap e\neq\emptyset} p_{t}(e') < 2$. These are indeed good rounds for all of nodes $\{v_1, v_2, \dots, v_d\}$. But we capped the number of good rounds for each node to $\Theta(dK^2 \log 1/\delta)$. This is a contradiction, if we choose the constant in $T$ large enough.
\end{proof}

Now, we are ready to put together these pieces and present our $(1+\eps)$-approximation:
\begin{theorem} \label{thm:MCM1LOCAL} There is a distributed algorithm in the $\mathsf{LOCAL}$ model that computes a $1+\eps$ approximation of maximum unweighted matching in $O(\frac{\log\Delta} {\log\log \Delta})$ rounds, for any constant $\eps>0$.
\end{theorem}
\begin{proof}We have $O(1/\eps)$ many phases, in each of which we find a nearly-maximal independent set of augmenting paths of length $d$ for $d=1, 2, \dots, O(1/\eps)$. Each phase takes $\Theta(d^2K^3\log 1/\delta + d^2\log_{K}{\Delta})$ rounds. We will set $\delta = \Theta(\eps^2)$ and $K=\log^{0.1} \Delta$. Hence, this is a complexity of $O(\frac{\log\Delta} {\eps^2\log\log \Delta})$ per phase, and thus at most $O(\frac{\log\Delta} {\eps^3\log\log \Delta})$ overall, which is $O(\frac{\log\Delta} {\log\log \Delta})$ for any constant $\eps>0$.

 In each of these phases, each node might get deactivated with probability at most $\delta$. Hence, the overall probability of a node becoming deactivated is at most $\delta'=O(\frac{\delta}{\eps})$. As \Cref{lem:MaximalityAmongActives} implies, in each phase, the found set of augmenting paths of length $d$ is indeed maximal in the graph induced by active nodes. Hence, after $O(1/\eps)$ phases, there is no augmenting path of length less than $O(1/\eps)$ among the active nodes. Thus, the matching at that point is a $(1+\eps/2)$-approximation of the maximum matching in the graph induced by active nodes.
Throughout the iterations, we deactivate each node with probability $\delta'=O(\delta/\eps)$. Considering the optimal matching OPT, the expected number of matching edges of OPT that we remove by deactivating their nodes is at most $2\delta'|OPT|$.\footnote{We indeed have also a concentration around this expectation, especially if $\Delta\leq n^{\gamma}$ for an appropriately chosen constant $\gamma$, then the statement holds with high probability, modulo a 2 factor.} On the remaining nodes, the matching we have found is a $(1+\eps/2)$-approximation of the maximum matching. Hence, overall, the found matching is a $(1+\eps/2)(1+2\delta')$-approximation. Setting $\delta = \Theta(\eps^2)$, this is a $(1+\eps)$-approximation.
\end{proof}
\subsection{A fast $(1+\eps)$-approximation of maximum cardinality matching in $\mathsf{CONGEST}$}
\label{sec:1epsMMCONGEST}
Here we extend a suitably modified variant of the algorithm of the previous subsection to the $\mathsf{CONGEST}$ model. This will provide a $(1+\eps)$-approximation of maximum cardinality matching in $O(\frac{\log \Delta}{\log\log \Delta})$ rounds of the $\mathsf{CONGEST}$ model, for any constant $\eps>0$. The key component will be a $\mathsf{CONGEST}$-model algorithm for computing a nearly-maximal independent set of augmenting paths of length at most $O(1/\eps)$ in bipartite graphs. We then utilize a method of Lotker et al.\cite{lotkerMatchingImproved} to use $2^{O(1/\eps)}$ iterations of applying this component to obtain a $(1+\eps)$-approximation for maximum cardinality matching in general graphs.

\paragraph{Finding Augmenting Paths in Bipartite Graphs.}
Bipartite graphs provide a nice structure for augmenting paths, which facilitates algorithmic approaches for computing them. This was first observed and used by Hopcroft and Karp\cite{HopcroftKarp1973} and was later used also by others, including Lotker et al.\cite{lotkerMatchingImproved} in the distributed setting. Consider a bipartite graph $H=(A, B, E)$ and a matching $M\subseteq E$ in it. Imagine orienting all matching $M$ edges from $A$ to $B$ and all others from $B$ to $A$. Then, a path $P$ is an augmenting path for $M$ if and only if it is a directed path starting in an unmatched $A$-node and ending in an unmatched $B$-node. See \Cref{fig:BipartiteAugmentation}. The three green dashed paths show three augmenting paths.

Considering the bipartite graph $H$ and matching $M$, we will find a nearly-maximal set of augmenting paths of length at most $O(1/\eps)$ as follows. We go over various possible lengths one by one. For each odd $d= 1, 3, \dots, L=O(1/\eps)$, we find a set of augmenting paths of length $d$ and then deactivate the remaining nodes who have augmenting paths of length $d$. This deactivation will be done such that overall the probability of deactivating each node is small. We then move on to augmenting paths of $d+2$, and so on, until length $L=O(1/\eps)$. This way, each time we are looking for augmenting paths of the shortest length; this helps in the computation.

We now explain our method for computing a near-maximal set of augmenting paths of length $d$, in bipartite graphs where the shortest augmenting path has length (at least) $d$. We will not construct the conflict graph between the augmenting paths explicitly. Instead, we will try to emulate the algorithm on the fly by means of simple communications on the base graph. In fact, the emulation will not be truthful and the implementation will deviate from the ideal $\mathsf{LOCAL}$-model algorithm, described in the previous subsection. We will however show that these alterations do not affect the performance measures significantly.

As before, for each short augmenting path $\mathcal{P}$ of length $d$, we will have a marking probability $p_{t}(\mathcal{P})$, for each iteration $t$. However, we will not insist on any one node knowing this probability; rather it will be known in a distributed manner such that we can still perform our necessary computations on it. For instance, a key part will be for each node $v$ to learn the summation of the probabilities $p_{t}(\mathcal{P})$ for augmenting paths $\mathcal{P}$ of length $d$ that go through $v$.

Let us first explain our distributed representation of $p_{t}(\mathcal{P})$. We have a time-variable attenuation parameter $\alpha_{t}(v)$ for each node $v$ that is either an $A$-node or an unmatched $B$-node. We will change these attenuation parameters from iteration to iteration. For simplicity, let us extend the definition also to matched nodes in $B$ but keep in mind that for each such node $v$, we will always have $\alpha_{t}(v)=1$. These attenuation parameters determine the marking probabilities of the augmenting paths as follows: for each augmenting path $\mathcal{P}$ of length $d$, its marking probability is the multiplication of the attenuation parameters along the path, that is $p_{t}(\mathcal{P})=\prod_{v\in \mathcal{P}} \alpha_{t}(v)$. We will later explain how we adjust these attenuation parameters over time.

We first explain how each node $v$ can learn the summation of the probabilities $p_{t}(\mathcal{P})$ for augmenting paths $\mathcal{P}$ of length $d$ that go through $v$. For simplicity, let us first assume that $p_{t}(\mathcal{P})$ is the same for all these augmenting paths, and thus we only need to figure out the number of augmenting paths of length $d$ that go through each node $v$.

\begin{figure}[t]
	\centering
		\includegraphics[width=0.8\textwidth]{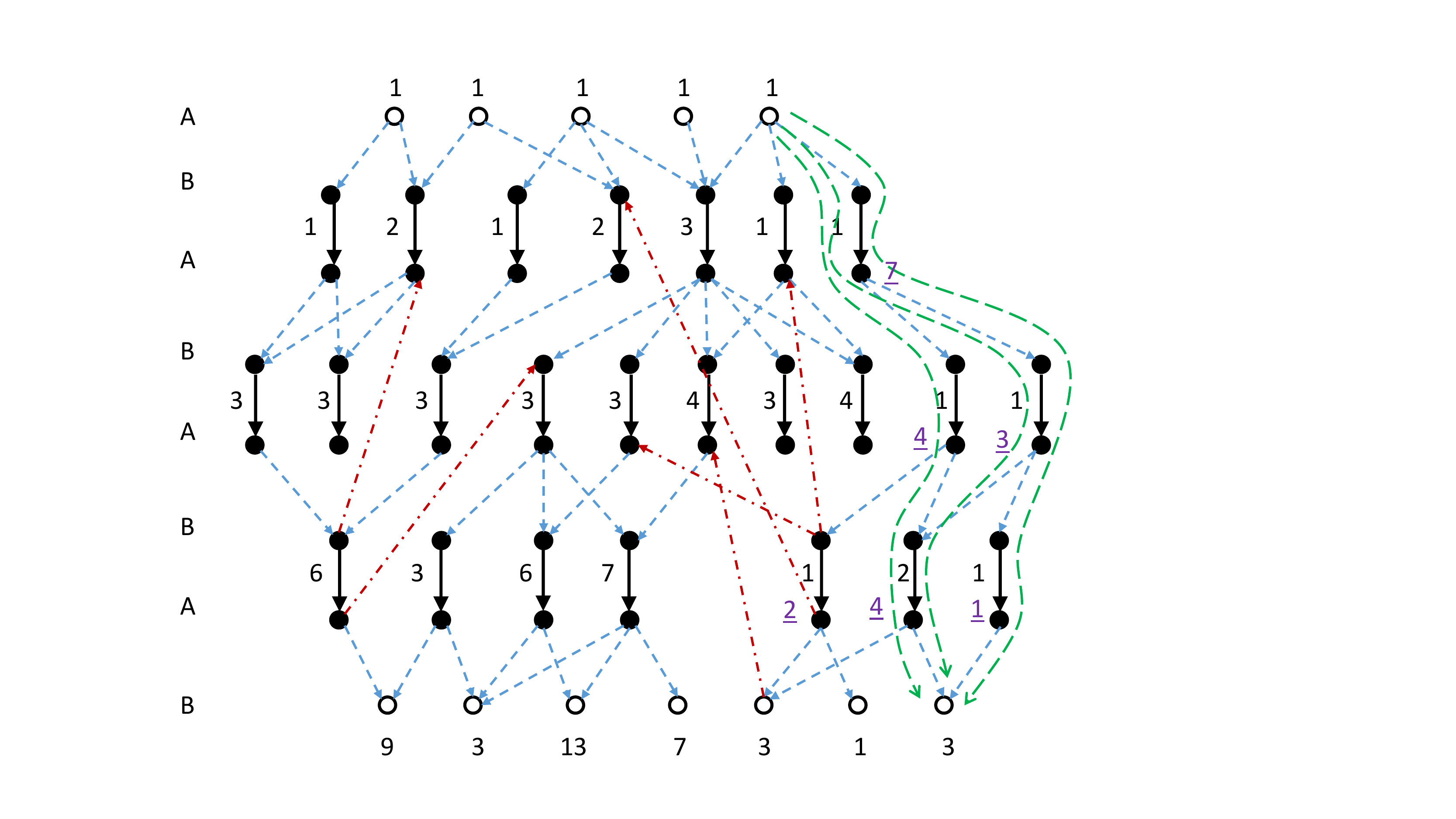}
	\caption{Augmenting Paths in a Bipartite Graph}
	\label{fig:BipartiteAugmentation}
\end{figure}

We will leverage the nice structure of the shortest augmenting paths here. This point was first observed and used by Lotker et a.\cite{lotkerMatchingImproved}. They show that a clean breadth first search style of message passing allows one to compute these numbers. The method is as follows. The reader might find \Cref{fig:BipartiteAugmentation} helpful while following this discussion. Notice that we are looking only for augmenting paths of length exactly $d$. Each unmatched $A$-node starts a counter $1$. Then, the algorithm proceeds for $d$ rounds, to find augmenting paths of length $d$ hop by hop. In each odd round, each $A$ node passes its count to all its $B$-neighbors. In each even round $t$, each matched $B$ node $b$ passes the summation of the numbers it received from its $A$-neighbors in the previous round $t-1$ to its matching-mate in $A$. Moreover, $b$ does this only in one round $t$ as any later round indicates augmenting paths with length greater than $d$.
If $t=d$, each unmatched $B$-node keeps the summation to itself without passing it on. Using an induction on the hop-counter, one can show that at the end, each unmatched node $b\in B$ learns the number of the augmenting paths of length $d$ that end at $b$. In \Cref{fig:BipartiteAugmentation}, the black numbers next to unmatched nodes or to matching edges indicate the numbers passed along during this forward traversal of augmenting paths. The red arrows indicate edges that are not a part of any shortest augmenting path, as in the BFS layers, they go back from a deeper layer to a shallower layer. These are edges for which the message would be sent after the receiving endpoint has already received summations in the previous rounds.

We next perform a similar procedure to ensure that all nodes know the respective number of short augmenting paths that go through them. This is by a simple time-reversal of these $d$ rounds of message passing. Start from each unmatched $B$-node $b$, and send back to each of its $A$-neighbors $a$ the number that $b$ received from $a$. Alternatively, this can be phrased as splitting the number that $b$ holds among its $A$-neighbors proportional to the numbers that it received from them during the forward traversal. Then, in each even round, each $A$-node passes to its matching mate in $B$ the summation of the numbers it received in the previous round from its $B$-neighbors. In each odd round, each $B$-node $b$ that received a number $x$ from its $A$-mate in the previous round splits this number $x$ among its $A$-neighbors $a$ proportional to the numbers that $b$ received from various $a$ during the corresponding round of the forward traversal of augmenting paths. After $d$ rounds, we reach back to the starting points of the augmenting paths;  at that time, each unmatched $A$-node keeps the received summation to itself without passing it on.

In \Cref{fig:BipartiteAugmentation}, the black numbers show the numbers passed on during the forward traversal, and the purple underlined numbers on the right side show the numbers sent during the backwards traversal for a few of the nodes.

\begin{claim}\label{clm:numberTraversals}
The number each node $v$ receives during the backwards traversal is equal to the number of augmenting paths that go through the node $v$.
\end{claim}
\begin{proof}We first consider the forward traversal. 

All a directed path a half-augmenting path if it starts at an unmatched $A$ node and alternates between unmatched and matched edges. We prove by induction on time $t\leq d$ that the number that each node $v$ receives in round $t$ is the number of (shortest) half-augmenting paths of length $t$ that end at $v$. The base case $t=0$ is trivial as each unmatched $A$ node starts with number $1$ and other nodes have $0$. In odd rounds $t$, node $v$ receives its first message in this round only if $v\in A$. Then, it receives the message from its matching mate $b\in B$, and by induction, the number that $b$ passed along is exactly the number of half-augmenting paths of length $t-1$ ending at $b$. Since adding the edge $v-b$ to each of these paths extends them by one hop, the number that $v$ receives is also the number of half-augmenting paths of length $t$ ending at $a$. In even rounds $t$, node $v$ receives in this round only if $v\in B$. Then, it received numbers from its $A$-neighbors, each of them indicating the number of half-augmenting paths of length $t-1$ that end at each neighbor $a\in A$. Each of these half-augmenting paths can be extended with the edge $a-b$, thus generating a half-augmenting paths of length $t$ ending at $v$. Hence, the number that $v$ received is again indeed the number of shortest half-augmenting paths of length $t$ ending at $v$. 

At the end, note that for $t=d$, we only look at the numbers received by unmatched $B$-node $b$ in round $d$. By the above inductive argument, the number received by $b$ is indeed the number of half-augmenting paths of length $d$ ending at $b$. But each of these is actually an augmenting path, as it ends in an unmatched $B$-node. Hence, we know that all unmatched $B$-nodes learn the number of augmenting paths of length $d$ ending at them.

We now consider the backwards traversal, and show by an induction on time $t$ in the backwards traversal that the number that each node $v$ receives in round $t$ of the backwards traversal is indeed the number of augmenting paths of length $d$ that go through $v$. The induction base $t=0$ follows from the argument above, as these are unmatched $B$-nodes which are the supposed endpoints of augmenting paths.
For the inductive step, we again have two cases: 

The case of even rounds follows easily because the number of augmenting paths going through a matched $B$-node $b$ is the same as the number of them that afterwards (in following the direction of the path) go through its mate $a\in A$. But this number is known to $a$ in the previous round of the backwards traversal, by induction. Hence, $b$ also learns its number of shortest augmenting paths.

For the case of odd rounds, let us examine an $A$-vertex $a$. During the forward traversal, in the corresponding round, the number that $a$ sent to each of its $B$-neighbors $b$ was the number of half-augmenting paths ending at $a$. Each neighbor $b$ now knows the number of shortest augmenting paths that continue by going through $b$. This might be a collection of paths, where various fractions come from various prior nodes $a\in A$. But the number for each prior node $a\in A$ is exactly proportional to the number of half-augmenting paths that end in $a$. Hence, when $b$ splits its number proportionally among its $a$-neighbors (proportional to the numbers received in forward traversal), each neighbor $a$ will learn the number of shortest augmenting paths that reach $a$, go from $a$ to $b$, and then from $b$ all the way to an unmatched $B$-node. Hence, node $a$ can just sum up all the numbers received from various $B$-neighbors $b$, and know the number of shortest augmenting paths going through $a$.
\end{proof}

Now we explain how a simple change in the previous message forwarding method makes each node know the probability summation of the paths that go through it. During the forward traversal, each unmatched $A$-node---which is a potential starting point of augmenting paths---passes its attenuation to its $B$-neighbors, instead of a passing a fixed number $1$. Then, each matched $B$-node $b$ passes the received summation to its $A$-mate $a$. The matching mate $a$ then attenuates this received summation via multiplying it by $\alpha_{t}(a)$, and then passes it on to its $B$-neighbors. Each unmatched $B$-node just applies its attenuation parameter on the received summation and keeps the summation to itself. The backwards traversal is essentially as before (without reapplying attenuations): each $B$-node splits its number among its $A$-neighbors proportional to the received numbers during the forward traversal. Each matched $A$ just passes the number to its matching mate in $B$.

\begin{claim}
The number that each node $v$ receives during the backwards traversal is equal to the summation of the $p_{t}(\mathcal{P})$ for all augmenting paths $\mathcal{P}$ that go through the node $v$.
\end{claim}
\begin{proof}[Proof Sketch] The proof follows by repeating the argument of \Cref{clm:numberTraversals}, and noting that along the path, the numbers are multiplied by the respective attenuations.
\end{proof}

\paragraph{Adjusting Attenuation Parameters Over Time.} Initially, we set $\alpha_{0}(v)=1/K$ for each unmatched $A$-node, which is the starting node of the potential augmenting paths, and $\alpha_{0}(v)=1$ for each matched $A$-node or unmatched $B$-node, which are the potential middle or end nodes of the augmenting paths. The updates are as follows: For each node $v$, if $\sum_{v, v\in \mathcal{P}} p_{t}(\mathcal{P}) \geq \frac{1}{10d}$, then set $\alpha_{t+1}(v) = \max\{\alpha_{t}(v) \cdot (\frac{1}{K})^{2d}, \Delta^{-20/\eps}\}$, and otherwise, $\alpha_{t+1}(v) = \min\{\alpha_{0}(v), \alpha_{t}(v) \cdot K\}$.

\paragraph{Remark About the Attenuation Lower Bound and the Floating-Point Precisions.} Notice that we have set a lower bound of $\Delta^{-20/\eps}$ on the attenuations. We say a node $v$ is \emph{stuck to the bottom} if $\sum_{v, v\in \mathcal{P}} p_{t}(\mathcal{P}) \geq \frac{1}{10d}$ which means $v$ wishes to lower its attenuation, but we already have $\alpha_{t}(v)=\Delta^{-20/\eps}$ and thus it cannot go further down. As we will see, this limitation has no significant effect on the behavior of the algorithm, because over all the at most $\Delta^{d} \ll \Delta^{-20/\eps}$ shortest augmenting paths going through a node, those that are stuck to the bottom make only a very negligible fraction of the probability. However, this lower bound allows us to keep the message size necessary for passing around attenuation values small. Particularly, with this lower bound, each attenuation can fit $O(\log \Delta/\eps)$ bits of precision. Thus, even when multiplied over paths of length $O(1/\eps)$, and summed up over at most $\Delta^{O(1/\eps)}$ paths, we only need $O(\log \Delta/\eps^2)$ bits of precision to represent the resulting number. This certainly fits $O(1/\eps^2)$ messages of the standard $\mathsf{CONGEST}$ model. Hence, by grouping each $\Theta(1/\eps^2)$ consequent rounds of the $\mathsf{CONGEST}$ model and treating them as one round, we have enough space for the desired precision.

\begin{definition} Call a node $v$ heavy in a given iteration $t$ if $\sum_{v, v\in \mathcal{P}} p_{t}(\mathcal{P}) \geq \frac{1}{10d}$. Moreover, call an augmenting path $\mathcal{P}$ heavy in that iteration if it goes through at least one heavy node.
\end{definition}

\begin{claim}\label{clm:ShrinkageGrowthBounds}  Each heavy augmenting path of length $d$ will have its $p_{t}(\mathcal{P})$ decrease by a factor in $[(\frac{1}{K})^{2d^2}, (\frac{1}{K})^{d}]$, unless at least one of its heavy nodes is stuck to the bottom, in which case the path's probability will remain at most $\Delta^{-20/\eps}$. Each non-heavy augmenting path of length $d$ will have its $p_{t}(\mathcal{P})$ increase by a factor in $[K, K^{d}]$, unless $p_{t}(\mathcal{P})=1/K$ in which case it will remain there and we have $p_{t+1}(\mathcal{P})=1/K$.
\end{claim}
\begin{proof} If a path is heavy, at least one of its nodes is heavy and then that node multiplies its own attenuation parameter by $(\frac{1}{K})^{2d}$, unless its attenuation is stuck to the bottom, in which case it remains there at $\Delta^{-20/\eps}$. The other $d-1$ nodes might also shrink their attenuation factor similarly or they might raise it by a $K$ factor, up to at most their original setting. In the lowest extreme, the overall multiplication of attenuation factors goes down by a  $(\frac{1}{K})^{2d^2}$ factor. In the highest extreme, $d-1$ of them raise their factor by a $K$ which means that we still have a shrinkage of $(\frac{1}{K})^{2d} \cdot K^{d-1} < (\frac{1}{K})^{d}$, if none of the heavy nodes is stuck to the bottom. If there is at least one such stuck heavy node, it will keep its attenuation at $\Delta^{-20/\eps}$, which means the overall marking probability of the path is at most $\Delta^{-20/\eps}$.

Now suppose that the path is not heavy. On the highest extreme, all $d$ nodes raise their attenuation parameters by a $K$ factor, which would mean an overall increase of a $K^{d}$ factor. On the lowest extreme, either the path already has all its attenuation parameters set as in the beginning, hence has $p_{t}(\mathcal{P})=1/K$, or at least one of its nodes raises its attenuation by a $K$ factor. Since the path is not-heavy, by definition, there is no heavy node on it and thus none of the nodes will reduce its attenuation parameter. Therefore, in this case, the marking probability raises by a $K$ factor, unless it is already at $1/K$, in which case it stays the same.
\end{proof}

\paragraph{The Algorithm for Marking Augmenting Paths.} We now explain how to use the marking probabilities $p_{t}(\mathcal{P})$, which are maintained implicitly by means of attenutation parameters, to mark augmenting paths and find a large independent set of them. Each free $B$-node $b$ has a summation $z = \sum_{\mathcal{P}, b\in \mathcal{P}} p_{t}(\mathcal{P})$ of all augmenting paths $\mathcal{P}$ that end at $b$. If the summation $z$ is greater than $1/d$, which means node $b$ is heavy, node $b$ will not initiate any path marking. Otherwise, node $b$ tosses a coin which comes out head with probability $z$, and if head, then it initiates an augmenting path marking operation. In this case, node $b$ passes the path marking token to one of its matched $A$-neighbors, as in the backwards traversal, chosen with a probability proportional to the sums that $b$ received from those neighbors during the forward traversal. If two marking tokens arrive at the same node, they both die. If a token is the only one that has arrived at a matched node $a'$, it will be passed to the matching mate $b'$ of $a'$. Then, $b'$ passes this token to one of its $A$-neighbors, again chosen with probabilities proportional to the sums that $b'$ received from those neighbors during the forward traversal. After continuing this process for $d$ iterations, some tokens might arrive at unmatched $A$-nodes. These are the marked paths that do not have any marked intersecting paths, they will be added to the independent set of short augmenting paths.

Note that instead of this stochastic link-by-link sampling of the marked path, one could have imagined a one-short process of sampling the marked path. In that process, the endpoint had sampled this full path at the beginning (if it knew the whole topology), we just pass the token along the path, and then the path is maintained if no intersecting path was marked. Indeed, one can easily see that the former process can only generate a larger set of isolated marked paths.

The tokens that make it through all the way and reach to an unmatched $A$ node are successful tokens. To announce that to all nodes along the path, these tokens will reverse their direction of traversal and go back all the way to their origin in the unmatched $B$-node. This can be done easily by each node remembering where from it received each of the tokens. While doing that, the token deactivates all nodes of the path, removing them from the problem. This also effectively augments $M$ with this path, by erasing the matching edges of the path, and substituting them with the unmatched edges.

\begin{claim}\label{clm:pathMarking} Each non-heavy path $\mathcal{P}$ gets marked and removed with probability at least $9p_{t}(\mathcal{P})/10$.
\end{claim}
\begin{proof} As stated before, we can think of the link by link creation of the marked path as a one-shot sampling by the path endpoint $b$ which is an unmatched $B$-node. Then the marking gets erased if there is another marked path intersecting it. For the non-heavy path $\mathcal{P}$, the probability that it gets an initial marking in this sense is exactly $p_{t}(\mathcal{P})$, as its unmatched endpoint $b$ will toss a coin with probability equal to the summation of all paths that end at $b$, because it is not heavy, and then pass this token backwards for a link by link creation. We next argue that conditioned on the initial marking, the path has a decent probability of not having any other intersecting path get marked. Since the probability of each other path $\mathcal{P}'\neq \mathcal{P}$ getting an initial marking is at most $p_{t}(\mathcal{P}')$, for each node $v$, the probability that $v$ receives a marking token through any of the paths $\mathcal{P}'\neq \mathcal{P}$ is at most $\sum_{v, v\in \mathcal{P}'} p_{t}(\mathcal{P}') \leq \frac{1}{10d}$, simply by a union bound. Thus, a union bound over all the at most $d$ nodes of the path $\mathcal{P}$ shows that with probability at least $1-\frac{d}{10d}$, none of them has any other marked intersecting path and thus, the $\mathcal{P}$ retains its mark and gets removed.
\end{proof}

Similar to \Cref{sec:1epsMM}, we call a node $v$ \emph{good} in a marking iteration $t$ if the summation of the probabilities of the light paths going through node $v$ is large. Here, the particular definition will require the summation to be at least $1/(dK^{2d})$. If a node $v$ goes through $\Theta(dK^{2d} \log 1/\delta)$ good iterations without being removed, we manually remove $v$ from the problem, knowing that the probability of this event is at most $\delta$.

\begin{lemma} The probability that a node $v$ does not get removed during $\Theta(dK^{2d} \log 1/\delta)$ good iterations is at most $\delta$.
\end{lemma}
\begin{proof}
Consider one iteration $t$ that is good for node $v$. We claim that in this iteration, node $v$ gets removed from the problem with probability at least $9/(10dK^{2d})$. This is because, as \Cref{clm:pathMarking} shows, for each light path $\mathcal{P}$ that goes through $v$, there is a $9p_{t}(\mathcal{P})/10$ probability that this path is marked for removal while no other path intersecting it is marked. Since these events are disjoint for different light paths $\mathcal{P}$ going through $v$, we can say that the probability of $v$ being removed because one of those light paths was removed is at least $9/(10dK^{2d})$. Now, the probability that this does not happen in the course of $\Theta(dK^{2d} \log 1/\delta)$ good iterations is at most $\delta$.
\end{proof}

Gathering probability sums of light paths and thus identifying good nodes can be done easily by a repetition of the previous forward and backwards probability passing processes, but this time just letting them pass only through light nodes.

\begin{lemma} After $T=O(\frac{d^4\log \Delta}{\log\log \Delta})$ iterations, there is no augmenting path of length $d$ remaining.
\end{lemma}
\begin{proof} We consider an arbitrary augmenting path $\mathcal{P}=(v_1, v_2, \dots, v_d)$ and prove that it cannot be the case that all of its nodes remain active for $T=\Theta(d^4K^{2d}\log 1/\delta + d^3\log_{K}{\Delta})$ iterations. Setting $K=\log^{1/(3d)} \Delta$, this is $O(d^4 \frac{\log \Delta}{\log\log \Delta})$ iterations for any $\delta\geq 2^{-\log^{0.3} \Delta}$. We assume that path $\mathcal{P}$ is not removed in the first $T=O(d^4 \frac{\log \Delta}{\log\log \Delta})$ iterations and we show that this leads to a contradiction, assuming a large enough constant in the asymptotic notation definition of $T$.  We emphasize that this is a deterministic guarantee, and it holds for every such augmenting path.

For a node $v \in \mathcal{P}$, if iteration $t$ is heavy but not good, then at most $1/(dK^{2d})$ weight in the summation can come from light paths going through $v$. That is at most a $\frac{1/(dK^{2d})}{1/d} =1/(K^{2d})$ fraction of the summation, as a heavy node has summation at least $1/d$.  Hence, ignoring the paths that their probability is stuck to the bottom, in every heavy but not-good iteration, the summation shrinks by a factor of $1/(K^{2d}) \cdot K^d + (1-1/(K^2)) \cdot (1/K)^d \leq 2/(K^d)$. At most $\Delta^{2/\eps} \cdot \Delta^{-20/\eps}$ weight can come from paths that their probability is stuck to the bottom bound of $\Delta^{-20/\eps}$. This is a total weight of $\Delta^{-18/\eps}$, which for a heavy node is at most $\frac{\Delta^{-18/\eps}}{1/d} \ll 1/K^d$ fraction of the overall weight.  Hence, even taking these paths into account, we see a shrinkage by a $3/K^d$ factor.

In each heavy and good iteration, the summation grows by at most a $K^d$ factor, which is in effect like canceling no more than $2$ of the shrinkage iterations. The number of good iterations for $v$ is capped to at most $\Theta(dK^{2d}\log 1/\delta)$, after which $v$ gets deactivated. Therefore, since $\sum_{e, v\in e} p_{t}(e)$ starts with a value of at most $\Delta^d/K$, and shrinks by a $3/K^{d}$ factor in every non-canceled heavy but not-good iteration, node $v$ can have at most $h=\Theta(dK^{2d}\log 1/\delta) + 3\log_{K}{\Delta}$ heavy iterations. Now this implies that the augmenting path $\mathcal{P}=(v_1, v_2, \dots, v_d)$ cannot have more than $h=\Theta(d^2K^{2d}\log 1/\delta) + 4d\log_{K}{\Delta}$ iterations in which it is heavy. This is because, in every heavy iteration for $\mathcal{P}$, at least one of its $d$ nodes must be heavy, and each of them has at most $\Theta(dK^{2d}\log 1/\delta) + 4\log_{K}{\Delta}$ heavy iterations.

As \Cref{clm:ShrinkageGrowthBounds} shows, during each iteration that path $\mathcal{P}$ is heavy, its probability shrinks by a $(\frac{1}{K})^{2d^2}$ factor, at worst (i.e., on the lower extreme). Each other iteration raises $p_{t}(\mathcal{P})$ by at least a $K$ factor, unless $p_{t}(e)$ is already equal to $1/K$. Since every $2d^2$ many of $K$-factor raises cancel one $(\frac{1}{K})^{2d^2}$-factor shrinkage, and as $p_{t}(e)$ starts at $1/K$, with the exception of $(2d^2+1)h$ iterations, all remaining iterations have $p_{t}(e)=1/K$. Among these, at most $h$ can be iterations in which the path is heavy. Thus, hyperedge $e$ has $T-(2d^2+2)h$ iterations in which $p_{t}(e)=1/K$ and it is not heavy. These are indeed good iterations for all of nodes $\{v_1, v_2, \dots, v_d\}$. But we capped the number of good iterations for each node to $\Theta(dK^{2d} \log 1/\delta)$. This is a contradiction, if we choose the constant in $T$ large enough.
\end{proof}

\begin{theorem} \label{thm:MCM1CONGEST} There is a distributed algorithm in the $\mathsf{CONGEST}$ model that computes a $(1+\eps)$-approximation of maximum unweighted matching in $O(\frac{\log\Delta} {\log\log \Delta})$ rounds, for any constant $\eps>0$.
\end{theorem}
\begin{proof}
As mentioned before, we utilize a method of Lotker et al.\cite{lotkerMatchingImproved} to compute a $1+\eps$ by finding (nearly-)maximal independent sets of short augmenting paths in bipartite graphs.

Lotker et al.\cite{lotkerMatchingImproved} explain a clever method for $(1+\eps)$-approximation of Maximum Cardinality Matching in general graphs by, in a sense, randomly transforming the problem into bipartite graphs. Concretely, there are $2^{O(1/\eps)}$ stages. In each stage, we get a bipartite graph and we need to find a (nearly-)maximal independent set of augmenting paths of length at most $2\lceil{1/\eps}\rceil-1$ in this bipartite graph. Then, we would augment the current matching with the found set of augmenting paths, and repeat. As Lotker et al. show, at the end, the found matching would be a $(1+\eps)$-approximation.

To find augmenting paths in biparite graphs, we use the method devised and explained above. Particularly, given the bipartite graph, we work on length values $d$= $1, 3, \dots, 2\lceil{1/\eps}\rceil-1$ one by one, in each length, we deactivate a small $\delta$ fraction of nodes and find an independent set of augmenting paths of length $d$ that is maximal among active nodes. This step takes $T=O(\frac{d^4\log \Delta}{\log\log \Delta})$ iterations, and each iteration can be implemented in $O(d/\eps^2)$ rounds, where the $d$-factor comes from the length of the path and the fact that we have to traverse it, and the $1/\eps^2$ factor comes from the fact that we need to send around messages that need $O(\log n/\eps^2)$ bits. Overall, this is at most $O(\frac{d^5\log \Delta}{\eps^2\log\log \Delta})$ rounds per iteration and thus at most $O(\frac{d^6\log \Delta}{\eps^2\log\log \Delta})$ rounds per stage, for going through all the possible length values. That is still $O(\mathsf{poly}(1/\eps) \cdot \frac{\log \Delta}{\log\log \Delta})$ rounds.

The method for generating the bipartite graphs is quite clean: color each node randomly red or blue, each with probability $1/2$. Then, keep each node in the bipartite graph if it is unmatched, or if it is matched but its matching edge becomes bi-chromatic. Also, keep all bi-chromatic edges supported on these nodes. This is clearly a bipartite graph and moreover, the step of creating this graph can be easily performed distributedly.

Each time that we use the nearly-maximal independent set algorithm, the guarantee is that, except for a negligible portion of at most $\delta |OPT|$ nodes which we excuse and deactivate, for a small $\delta \ll \eps$, the found augmenting paths are maximal in that bipartite graph among the remaining active nodes. Overall, we will deactivate each node with probability at most $\delta'=\delta 2^{O(1/\eps)}$. By choosing $\delta = 2^{-\Omega(1/\eps)}$, which increases the round complexity of the near-maximal augmenting path algorithm only by an $O(1/\eps)$ factor, we can ensure that $\delta'\ll \eps$. Thus, the overall approximation would remain at most $(1+\eps)(1+\delta') \leq 1+2\eps$.  Moreover, the round complexity over all stages is $O(2^{O(1/\eps)} \cdot \frac{\log \Delta}{\log\log \Delta})$.
\end{proof}

\subsection{Alternative Fast Method for $2+\eps$ Approximation of Unweighted Maximum Matching}
Here, we explain an alternative method for computing a $2+\eps$-approximation of maximum unweighted matching. This approach might be interesting especially because of its simplicity, but perhaps also because of the better probability concentration that it provides on the approximation.
\subsubsection{The Algorithm for Bipartite Unweighted Graphs}

\label{sec:2epsMMbipartite}
\paragraph{The Algorithm.} In each round, each node $v$ on the left side of the bipartite graph sends a \emph{matching proposal} on a randomly chosen one of its remaining edges to right side nodes. Each right side node accepts the proposal from the highest id, if there is one.

\begin{lemma} For any $K$, the algorithm produces a matching with approximation factor $2+\eps$, in $O(K\log 1/\eps+ \log \Delta/\log K)$ rounds, with high probability. Particularly, each left-node in the optimal matching remains unmatched but non-isolated with probability at most $\eps/2$. Optimizing over $K$ leads to round complexity $O(\log \Delta/\log(\log \Delta/\log(1/\eps)))$.
\end{lemma}

\begin{proof} Let us call a left node unlucky if it remains unmatched but non-isolated after $O(K\log 1/\eps+ \log \Delta/\log K)$ rounds. We first show that for each left node $v$, the probability that $v$ is unlucky is at most $\eps/2$. Moreover, this depends only on the randomness of $v$ and it holds regardless of the randomness used by other nodes.

The key part in that is to show that in each round, either $v$'s degree falls by a $K$ factor or its proposal succeeds with probability at least $1/K$. To prove this intermediate claim, let's examine the proposals of the left nodes in descending order of ids. When we reach the turn of node $v$, either less than $1/K$ fraction of $v$'s right neighbors remain unmatched, in which case $v$'s degree has fallen by a $K$ factor, or otherwise $v$ has a chance of at least $1/K$ that its proposal was to a a currently unmatched neighbor, in which case $v$ would be successfully matched.

After $O(K\log 1/\eps+ \log \Delta/\log K)$ rounds, the probability that $v$ remains unmatched but non-isolated is at most $\eps/2$. This is because, there can be at most  $\log \Delta/\log K$ rounds where degree falls by a $K$ factor and the probability of failing in each other round is at most $1-1/K$.

Let us now examine the approximation factor. Consider the OPT matching. Each left-node of OPT gets unlucky with probability at most $\eps/2$. Hence, at most $\eps$ fraction of its left nodes get unlucky, with probability\footnote{This is similar to the probability concentration obtained by Lotker et al.\cite{lotkerMatchingImproved}, which they call \emph{high probability}.} at least $1-e^{-\Omega(\eps |OPT|)}$. In the remainder, each edge of the found matching can kill at most $2$ edges of the OPT matching. So the found matching is a $(2+\eps)$-approximation, with high probability. 
\end{proof}

\subsubsection{The Algorithm for General Unweighted Graphs}
We solve the general case by randomly transforming it to the bipartite case. In each iteration, randomly call each node left or right with probability half. This produces a bipartite graph which \emph{preserves} each edge of the OPT matching with probability $1/2$. Run the bipartite algorithm on this graph and then remove the found matching edges, and all $G$-edges incident on their endpoints. Repeat this experiment $O(\log 1/\eps)$ times.

\begin{lemma} The algorithm finds a $(2+\eps)$-approximation of the maximum matching in general graphs in $O(\log 1/\eps \cdot \frac{\log \Delta}{\log(\log \Delta/\log(1/\eps))})$ rounds.
\end{lemma}
\begin{proof}
The time complexity is immediate as the algorithm is made of $O(\log 1/\eps)$ repetitions of the bipartite algorithm. We next present the approximation analysis.
Consider the OPT matching $\mathcal{M}$, and let us examine what happens to one edge $e=\{v, u\}$ where $e \in\mathcal{M}$. In the first iteration that either of nodes $v$ and $u$ is matched (not necessarily to each other), we charge the blame of $e$ not being in our matching to that found matching edge incident on $v$ or $u$ (to exactly one of them if there are two). This way, each found matching edge $e'$ might be blamed for at most $2$ matching edges of OPT, one from each endpoint of $e'$. Now, in each of the first $O(\log 1/\eps)$ iterations, either edge $e$ is already removed from the graph because at least one of its endpoints was matched before, or with probability $1/2$ edge $e$ gets preserved in the bipartite graph, in which case there are only two ways it is not taken into the matching: (1) one of its endpoints gets matched to a different node, (2) the left node endpoint got unlucky in the bipartite algorithm. Note that the latter happens with probability at most $\eps/4$, as the bipartite algorithm guarantees. Throughout the $O(\log 1/\eps)$ iterations, with probability $1-\eps/4$, edge $e\in \mathcal{M}$ is preserved in at least one iteration, unless it was already removed. Hence, we can say that with probability $1-\eps/2$, edge $e \in \mathcal{M}$ has put its blame on some other edge added to the matching. Note that we also have independence between different matching edges edge in $\mathcal{M}$, as (1) whether they get unlucky in the bipartite algorithm is independent, and (2) whether they get preserved in a bipartite transformation is independent. Thus, by Chernoff bound, at most $\eps$ fraction of OPT edges have not put their blame on some edge in our found matching, with probability at least $1-e^{-\Omega(\eps|OPT|)}$. Since each found matching edge is blamed at most twice, the found matching is a $(2+\eps)$-approximation, with high probability.
\end{proof}

\end{document}